\documentclass[11pt]{article}
\usepackage[hmargin=1in,vmargin=1in]{geometry}
\usepackage{hchang}
\usepackage{cleveref}
\usepackage{comment}
\usepackage{bm}

\nolinenumbers

\newtheorem{theorem}{Theorem}[section]
\newtheorem{lemma}[theorem]{Lemma}
\newtheorem{claim}[theorem]{Claim}
\newtheorem{observation}[theorem]{Observation}

\newtheorem{definition}[theorem]{Definition}

\def\eps{\e}
\newcommand{\score}{\ensuremath{\operatorname{score}}}
\def\dist{\delta}
\newcommand{\ddim}{\ensuremath{\mathsf{ddim}}}

\title{Optimal Euclidean Tree Covers}
\author{%
Hsien-Chih Chang%
\thanks{Department of Computer Science, Dartmouth College. Email: {\tt hsien-chih.chang@dartmouth.edu}.} 
\and 
Jonathan Conroy%
\thanks{Department of Computer Science, Dartmouth College. Email: {\tt jonathan.conroy.gr@dartmouth.edu}}  
\and 
Hung Le%
\thanks{Manning CICS, UMass Amherst. Email: {\tt hungle@cs.umass.edu}}  
\and
Lazar Milenkovic%
\thanks{Tel Aviv University. Email: {\tt milenkovic.lazar@gmail.com}}  
\and
Shay Solomon%
\thanks{Tel Aviv University. Email: {\tt solo.shay@gmail.com}}  
\and
Cuong Than%
\thanks{Manning CICS, UMass Amherst. Email: {\tt cthan@cs.umass.edu}}  
}

\date{}

\begin{document}

\maketitle

\begin{abstract}
A \emph{$(1+\eps)$-stretch tree cover} of a metric space is a collection of trees, where every pair of points has a
$(1+\eps)$-stretch path in one of the trees.
The celebrated \emph{Dumbbell Theorem} [Arya \etal\ STOC'95] states that any set of $n$ points in $d$-dimensional Euclidean space admits a $(1+\eps)$-stretch tree cover with $O_d(\eps^{-d} \cdot \log(1/\eps))$ trees, where the $O_d$ notation suppresses terms that depend solely on the dimension~$d$.
The running time of their construction is $O_d(n \log n \cdot \frac{\log(1/\eps)}{\eps^{d}} + n \cdot \eps^{-2d})$. 
Since the same point may occur in multiple levels of the tree, the \emph{maximum degree} of a point in the tree cover may be as large as $\Omega(\log \Phi)$, where $\Phi$ is the aspect ratio of the input point set.

In this work we present a  $(1+\eps)$-stretch tree cover with $O_d(\eps^{-d+1} \cdot \log(1/\eps))$ trees, which is optimal (up to the $\log(1/\eps)$ factor). 
Moreover, the maximum degree of points in any tree is an \emph{absolute constant} for any $d$.  
As a direct corollary, we obtain an optimal {routing scheme} in low-dimensional Euclidean spaces. 
We also present a 
$(1+\eps)$-stretch \emph{Steiner} tree cover (that may use  Steiner points) with $O_d(\eps^{(-d+1)/{2}} \cdot \log(1/\eps))$ trees, which too is optimal. The running time of our two constructions is linear in the number of edges in the respective tree covers, ignoring an additive $O_d(n \log n)$ term; this improves over the running time underlying the Dumbbell Theorem.
\end{abstract}

\section{Introduction}

Let $M$ be a given metric space with distance function $\dist$, and $X$ be a finite set of points in~$M$.
A \EMPH{tree cover} for $(M,X)$ is a collection of trees $\mathcal{F}$, each of which consists of (only) points in $X$ as vertices and abstract edges between vertices, such that between every two points $x$ and $y$ in $X$,  $\dist_M(x,y) \le \dist_T(x,y)$ for every tree $T$ in $\mathcal{F}$.
A tree cover $\mathcal{F}$ has \EMPH{stretch~$\alpha$} if for every two points $x$ and $y$ in $X$, there is a tree $T$ in $\mathcal{F}$ that preserves the distance between $x$ and $y$ up to $\alpha$ factor: $\dist_T(x,y) \le \alpha\cdot\dist_M(x,y)$.
We call such $\mathcal{F}$ an \EMPH{$\alpha$-tree cover} of $X$.
In this paper, we will focus on the scenario where $M$ is the $d$-dimensional Euclidean space for some constant $d = O(1)$.  It is not hard to see that, in this case, the edges can be drawn as line segments in $\R^d$ between the corresponding two endpoints, with weights equal to their Euclidean distances.
If we relax the condition so that trees in $\mathcal{F}$ may have other points from~$M$ (called \EMPH{Steiner points}) as vertices instead of just points from $X$, the resulting tree cover is called a \EMPH{Steiner tree cover}. 

Constructions of tree covers, due to their algorithmic significance, are subject to growing research attention~\cite{AP92, AKPW95, ADM+95, GKR01, MN06, DYL06, CGMZ16, BFN22, CCL+23a, CCL+23b}; by now generalizations in various metric spaces and graphs are well-explored.
The main measure of quality for tree cover is its \emph{size}, that is, the number of trees in a tree cover $\mathcal{F}$.
The existence of a small tree cover provides a framework to solve distance-related problems by essentially reducing them to trees. 
Exemplified applications include distance oracles~\cite{CCL+23a,CCL+23b}, labeling and routing schemes~\cite{TZ05,KLMS22}, spanners with small hop diameters~\cite{KLMS22}, and bipartite matching~\cite{ACRX22}.

The celebrated \emph{Dumbbell Theorem} by Arya, Das, Mount, Salowe, and Smid~\cite{ADM+95} from almost thirty years ago demonstrated that in $d$-dimensional Euclidean space, any point set $X$ has a tree cover of stretch $1+\e$ that uses only $O_d(\e^{-d} \cdot \log(1/\e))$ trees.\footnote{The $O_d$ notation suppresses terms that depend solely on the dimension $d$.}
Moreover, the tree cover can be computed within time $O_d \Paren{n \log n \cdot \frac{\log(1/\eps)}{\eps^{d}} + n \cdot \eps^{-2d} }$, where $n$ is the number of points in $X$.
In the Euclidean plane (when $d=2$), this gives us a tree cover of size $O(\e^{-2} \cdot \log(1/\eps))$.
The theorem has a long and complex proof, which spans a chapter in the book of Narasimhan and Smid~\cite{NS07}.
A few years ago, this theorem was generalized for doubling metrics%
\footnote{the \EMPH{doubling dimension} of a metric space $(M,\dist)$ is the smallest value $\ddim$ such that every ball in $M$ can be covered by $2^\ddim$ balls of half the radius;  
a metric $\dist$ is called \EMPH{doubling} if its doubling dimension is constant} 
by Bartal, Fandina, and Neiman~\cite{BFN22}, 
who achieved the same bound as \cite{ADM+95} via a much simpler construction; the running time of their construction was not analyzed.%
\footnote{In high-dimensional Euclidean spaces the upper bound in \cite{BFN22} improves over that of \cite{ADM+95}, since the $O_d$ notation in \cite{ADM+95} and \cite{BFN22} suppress multiplicative factors of
$d^{O(d)}$ and $2^{O(d)}$, respectively.}
In the constructions by \cite{ADM+95,BFN22}, 
same point may have multiple copies in different levels of the tree, hence the maximum degree of points%
\footnote{the \EMPH{degree} of a point is the number of edges incident to it} 
may be as large as $\Omega(\log \Phi)$, where $\Phi$ is the aspect ratio of the input point set;  
see Sections~\ref{S:technical} and \ref{sec:ConstDeg} for a more detailed discussion.

Since the number of trees provided by the two known constructions \cite{ADM+95,BFN22} matches the packing bound $\e^{-d}$ (up to a logarithmic factor), it is tempting to conjecture that this bound is tight.
However, there is a gap between this upper bound and the best lower bound we have, which comes indirectly from $(1+\eps)$-stretch \emph{spanners}. 
For any parameter $\alpha \ge 1$, a \EMPH{Euclidean $\alpha$-spanner} for any $d$-dimensional point set is a weighted graph spanning the input point set, whose edge weights are given by the Euclidean distances between the points, that approximates all the original pairwise Euclidean distances within a factor of $\alpha$.
We note that an $\alpha$-spanner can be obtained directly by taking the union of all trees in any $\alpha$-tree cover for the input point set.
The $\Omega(n \cdot \e^{-d+1})$ size lower bound for $(1+\e)$-spanners~\cite[Theorem~1.1]{LS22} directly implies that any $(1+\e)$-tree cover must contain $\Omega(\e^{-d+1})$ trees. 
This is an $\e^{-1}$-factor away from the packing bound.
In particular, in the Euclidean plane, there is a gap between the upper bound of $O(\eps^{-2})$ and the lower bound of $\Omega(\eps^{-1})$.
One can extend the notions of spanner by introducing {Steiner points} as well, which are additional points that are not part of the input.
A weaker $\Omega(\e^{(-d+1)/2})$ lower bound can be obtained for \emph{Steiner tree cover}, from the $\Omega(n/\sqrt{\e})$ size lower bound for Steiner $(1+\e)$-spanner in $\R^2$~\cite[Theorem~1.4]{LS22},
and the $\Omega(n/{\e}^{(d-1)/2})$ size lower bound in general $\R^d$~\cite{BT22}.

\subsection{Short Survey on Tree Covers}
\label{S:survey}

There are many papers published on tree covers in recent years, with subtle variations in their definitions due to differences in main objectives and applications.  
Here we attempt to summarize the best upper and lower bounds known to our knowledge, highlighting the tradeoff between tree cover size and stretch in the previous work.
Some of the bounds are not explicitly stated in the cited reference but can be deduced from it.
For additional relevant work, refer to \cite{BFN22} and the references therein.

\paragraph{General metrics.}
The earliest literature on the notion of tree cover is probably Awerbuch and Peleg~\cite{AP92} and Awerbuch, Kutten, and Peleg~\cite{AKP94}, focusing on graph metrics. 
Their main objective is to minimize the number of trees \emph{each vertex belongs to} (in the sparse cover sense) instead of minimize the total number of trees.
Thorup and Zwick~\cite[Corollary~4.4]{TZ05} improved over Awerbuch and Peleg~\cite{AP92} by constructing a Steiner tree cover with stretch $2k-1$ where every vertex belongs to $O(n^{1/\alpha} \cdot \log^{1-1/k} n)$ trees.
Charikar \etal~\cite{CCGGP98} studies a similar problem of probabilistically embedding finite metric space into a small number of trees.
Many of the earlier work on tree covers are motivated by application in routing~\cite{TZ01a}.

Gupta, Kumar, and Rastogi~\cite[Theorem~4.3]{GKR01} observed that any tree cover must have size $n^{\Omega(1/\alpha)}$ if the stretch is $\alpha$; the lower bound is based on the existence of girth-$g$ graphs with $n^{\Omega(1/g)}$ edges \cite{Mar88}\cite[Lemma~9]{ADDJS93}.
It is important to emphasize that the tree covers considered in \cite{GKR01} are \emph{spanning} --- the trees must be subgraphs of the input graph.
Bartal, Fandina, and Neiman established the same lower bound~\cite[Corollary~13]{BFN22} by reduction from spanners~\cite{ADDJS93}.
In a different direction, Dragan, Yan, and Lomonosov~\cite{DYL06} studied spanner tree covers with additive stretch on special classes of graphs, such as chordal graphs and co-comparability graphs.

One might relax the condition to allow vertices not presented in the graph (called \emph{Steiner} vertices) to be part of the tree cover.
By allowing Steiner vertices, 
Mendel and Naor~\cite{MN06} showed that 
any $n$-point metric space has a Steiner tree cover of size $O(\alpha\cdot n^{1/\alpha})$ and stretch $O(\alpha)$.
Bartal, Fandina, and Neiman~\cite{BFN22} obtained an inverse tradeoff:
any $n$-point metric space has a Steiner tree cover of size $k$ and stretch $O(n^{1/k}\cdot (\log n)^{1-1/k})$. 
In particular, this means we can get $O(\log n)$ trees with $O(\log n)$ stretch.
While the lower bound from \cite{GKR01} for spanning tree cover no longer holds when Steiner vertices are allowed, 
a similar lower bound of $\Omega(n^{4/(3\alpha+2)} / \log n) = n^{\Omega(1/\alpha)}$ for the size of Steiner $\alpha$-tree covers can be derived from Steiner spanners (also known as \emph{emulators})~\cite[Theorem~6]{ADDJS93},
using the same argument in \cite{BFN22}.

\paragraph{Doubling metrics.}
Chan, Gupta, Maggs, and Zhou~\cite[Lemma~3.4]{CGMZ16} constructs Steiner tree covers for doubling metrics~\cite{CGMZ16}.
More precisely, if the doubling dimension of the metric space is $d$,
their tree cover uses $O(d \log d)$ Steiner trees and has stretch $O(d^2)$.
Bartal, Fandina, and Neiman~\cite{BFN22} obtained two separated constructions:
one may have tree cover of stretch $O(\alpha)$ and $O(2^{d/\alpha} \cdot d \cdot \alpha)$ Steiner trees for any $\alpha \ge 2$~\cite[Theorem~7]{BFN22} using the $O(1)$-padded hierarchical partition family in \cite[Lemma~8]{ABN11}, or alternatively a tree cover with $(1+\e)$ stretch and $(1/\e)^{O(d)} \cdot \log(1/\e)$ trees~\cite[Theorem~3]{BFN22}
using net trees.
It is worth emphasizing that the second construction does not use Steiner points.
They also established a lower bound on the size of non-Steiner tree cover~\cite[Corollary~13]{BFN22} by reduction from spanners~\cite{ADDJS93}: there is an $n$-point metric space with doubling dimension $d$, such that any $\alpha$-tree cover requires $\Omega(2^{d/\alpha})$ trees.

\paragraph{Planar and minor-free graphs.}
On planar graphs Gupta, Kumar, and Rastogi~\cite{GKR01} constructed the first $O(\log n)$-size (non-Steiner) tree cover with stretch $3$.
Again this is improved by Bartal, Fandina, and Neiman in two different directions:
either one has stretch $O(1)$  and $O(1)$ trees~\cite[Corollary~9]{BFN22}
using the $O(1)$-padded hierarchical partition family in \cite{KLMN04}%
\footnote{the constants in \cite{KLMN04} imply that the stretch is at least $3^4 = 81$ and the number of trees is at least $3^3=27$}, 
or alternatively a tree cover with $(1+\e)$ stretch and $O(\e^{-1}\log^2 n)$ trees~\cite[Theorem~5]{BFN22}, using path separators~\cite{AG06}.
Their results naturally extend to minor-free graphs.
Recently, the authors 
get the best of both worlds by constructing a Steiner tree cover with $(1+\e)$ stretch using $\tilde{O}(1/\e^3)$ many trees~\cite{CCL+23a} through the introduction of a new graph partitioning scheme called the \emph{shortcut partition}; the result also extends to minor-free graphs~\cite{CCL+23b}.

On planar graphs $\Omega(\sqrt{n})$ trees are required if no stretch is allowed~\cite{GKR01}.  
However in the $(1+\e)$-stretch regime, we are not aware of any existing lower bounds.
The strongest lower bound for tree covers on planar graphs we managed to deduce comes from \emph{distance labeling}: 
Suppose we have a Steiner $(1+\e)$-tree cover using $O(\e^{-1/(3+\delta)})$ many trees for some $\delta > 0$. 
Then we can construct an approximate distance labeling scheme by concatenating the $O(\log n \cdot \log(1/\e))$-length labeling schemes for all trees~\cite{FGNW17}.  
By setting $\e = 1/n$, we get an exact labeling scheme for unweighted planar graph of length $\tilde{O}(n^{1/(3+\delta)})$, contradicting to an information-theoretical lower bound~\cite{GPPR04}.
This implies that any Steiner $(1+\e)$-tree covers on planar graphs requires at least $\tilde{\Omega}(n^{1/3})$ many trees.

\paragraph{Euclidean metrics.}

We already discussed tree cover results on Euclidean metrics in the introduction; here we mentioned a few additional facts.

All upper bound constructions on metrics with bounded doubling dimensions immediately apply to Euclidean metrics as well.
Surprisingly, relatively few lower bounds have been established in the literature for Euclidean spaces.
Early in the introduction we derived an $\Omega(1/\e^{d-1})$ lower bound for non-Steiner tree cover and an $\Omega(1/{\e}^{(d-1)/2})$ lower bound for Steiner tree cover in $\R^d$ by reduction from spanners.

One thing to notice is that in Euclidean spaces, the meaning of Steiner points differs slightly from its graph counterparts: 
after choosing a Steiner point (which lies in the ambient space $\R^d$), the weight of an edge incident to a Steiner point is determined by its Euclidean distance, unlike in the graph setting one may choose the weight freely (as the Steiner points are artificially inserted and were not part of the graph a priori).
One might think that such a distinction cannot possibly make any difference; 
however, recently Andoni and Zhang~\cite{AZ23} 
proved that $(1+\e)$-spanner of subquadratic size exists for arbitrary dimensional Euclidean space by allowing out-of-nowhere Steiner points,
while establishing lower bound simultaneously when the Steiner points are required to sit in the Euclidean space.  
They showed that there are $n$ points in $\R^d$ (for some highe dimension $d$ depending on $n$) where any $(\sqrt{2}-\e)$-spanner (with Euclidean Steiner points) requires $\Omega(\e^4 \cdot n^2 / \log^2 n)$ edges; the lower bound follows from a randomized construction and volume argument.
This translates to an almost linear lower bound ($\Omega(\e^4 \cdot n / \log^2 n)$ on the minimum number of trees required in any Euclidean Steiner tree cover with $(\sqrt{2}-\e)$ stretch.
All Steiner points used in our construction are Euclidean;
at the moment, we are unaware of any tree cover construction that obtains a better bound by taking advantage of the non-Euclidean Steiner points.

\paragraph{Ramsey trees.}
A stronger notion called the \emph{Ramsey} tree cover has been studied, where every vertex $x$ is associated with a tree $T_x$ in $\mathcal{F}$, such that the distance from $x$ to \emph{every} other vertex is approximated preserved by the same tree $T_x$.
Both the constructions of Mendel and Naor~\cite{MN06} and Bartal, Fandina, and Neiman~\cite{BFN22} for general metrics are indeed Ramsey trees. 
These bounds are essentially tight if the trees are required to be Ramsey; that is,
any Ramsey tree cover of stretch $\alpha$ must contain $n^{1/\Omega(\alpha)}$ many tree~\cite[Corollary~13]{BFN22}.
Even when the input metric is planar and doubling,
any Ramsey tree cover of stretch $\alpha$ must contain $n^{1/\Omega(\alpha\log\alpha)}$ many tree~\cite[Theorem~10]{BFN22}, and any Ramsey tree cover of size $k$ must has stretch $n^{1/k}$~\cite[Theorem~9]{BFN22}.

\subsection{Main Results}

We improve the longstanding bound on the number of trees for Euclidean tree cover by a factor of $1/\e$, for any constant-dimensional Euclidean space.\footnote{As with \cite{ADM+95}, the $O_d$ notation in our bound suppresses a multiplicative factor of $d^{O(d)}$, which should be compared to the multiplicative factor of
$O(1)^d$ suppressed in the bound of \cite{BFN22}. Thus, our results improve over that of \cite{BFN22} only under the assumption that $\eps$ is sufficiently small with respect to the dimension $d$; this assumption should be acceptable since the focus of this work, as with the great majority of the work on Euclidean spanners, is low-dimensional Euclidean spaces.}
 In view of the aforementioned lower bound \cite{LS22,BT22}, this is optimal up to the $\log(1/\eps)$ factor. 
Roughly speaking, we show that the packing bound barrier (incurred in both \cite{ADM+95} and \cite{BFN22}) can be replaced by the number of $\eps$-angled cones needed to partition $\R^d$; %
for more details, refer to Section~\ref{S:technical}.

\begin{theorem}
\label{thm:NonSteiner}
For every set of points in $\mathbb{R}^d$ and any $0 < \eps < 1/20$, there exists a tree cover with  stretch $1+\eps$ and $O_d(\eps^{-d+1}\cdot \log(1/\eps))$ trees. The running time of the construction is $O_d(n \log n + n \cdot \eps^{-d+1} \cdot \log(1/\eps))$.
\end{theorem}
We note our construction is faster than that of the Dumbbell Theorem of \cite{ADM+95} by more than a multiplicative factor of $\eps^{-d}$.

In addition, we demonstrate that the bound on the number of trees can be quadratically improved using \emph{Steiner points}; in $\R^2$ we can construct a Steiner tree cover with stretch $1+\eps$ using only $O(1/\sqrt{\eps})$ many trees.
The result generalizes for higher dimensions.
In view of the aforementioned lower bound \cite{LS22,BT22}, this result too is optimal up to the $\log(1/\eps)$ factor.

\begin{theorem}
\label{thm:Steiner}
For every set of points in $\mathbb{R}^d$ and any $0 < \eps < 1/20$, there exists a Steiner tree cover with stretch $1+\eps$ and $O_d(\eps^{(-d+1)/2} \cdot \log(1/\eps))$ trees.
 The running time of the construction is $O_d(n \log n + n \cdot \eps^{(-d+1)/2} \cdot \log(1/\eps))$.
\end{theorem}

\subsubsection{Bounded degree tree cover}

Although the number of trees in the tree cover is the most basic quality measure, together with the stretch, another important measure is the \emph{degree}. 
One can optimize the maximum degree of a point in any of the trees, or to optimize the maximum degree of a point over all trees --- both these measures are of theoretical and practical importance.

Both the Dumbbell Theorem~\cite{ADM+95} and the BFN construction~\cite{BFN22} use  \emph{copies} of the same point in multiple trees, and even in different levels of the same tree. 
Consequently, each point may have up to $\log\Phi$ copies, which can be viewed as distinct \emph{nodes} of the tree,
where $\Phi$ is the aspect ratio of the input point set. 
The Dumbbell trees have bounded \emph{node-}degree (which is improved to degree 3 in~\cite{Smid12}), but the maximum \emph{point-}degree in any tree could still be $\Theta(\log\Phi)$ after reidentifying all the copies of the points. 
The construction of~\cite{BFN22} may also incur a point-degree of $\Omega(\log\Phi)$ in any of the trees.\footnote{Even node-degrees may blow up in the construction of \cite{BFN22}, but it appears that a simple tweak of their construction can guarantee a node-degree of $\eps^{-O(d)}$.}

We strengthen  Theorem~\ref{thm:NonSteiner} by achieving a constant degree for each point in any of the trees; in fact, our bound is an absolute constant in any  dimension. 
As a result, the maximum degree of a point over all trees is $O_d(\eps^{-d+1} \cdot \log(1/\eps))$; this is optimal up to the $\log(1/\eps)$ factor, matching the average degree (or size) lower bound of spanners mentioned above \cite{LS22}. 

\paragraph{Routing.}

We highlight one application of our bounded degree tree cover to efficient routing.
\begin{theorem}\label{thm:routing}
For any set of points in $\mathbb{R}^d$ and any $0 < \eps < 1/20$, there is a compact routing scheme with stretch $1+\eps$ that uses routing tables and headers with $O_d(\eps^{-d+1} \log^2 (1/\e) \cdot \log{n})$ bits of space.
\end{theorem}
Our routing scheme uses smaller routing tables compared to the routing scheme of Gottlieb and Roditty~\cite{GR08Soda}, which uses routing tables of $O(\eps^{-d}\log{n})$ bits. 
At a high level, we provide an efficient reduction from the problem of routing in low-dimensional Euclidean spaces to that in trees; more specifically, we present a new labeling scheme for determining the right tree to route on in the tree cover of \Cref{thm:NonSteiner}.
Having determined the right tree to route on, our entire routing algorithm is carried out on that tree, while the routing algorithm of \cite{GR08Soda} is 
carried out on a spanner; routing in a tree is clearly advantageous over routing in a spanner, also from a practical perspective. 
Refer to 
\cite{GR08Soda} for the definition of the problem and relevant background.

\subsection{Technical Highlights}
\label{S:technical}

\subsubsection{Achieving an optimal bound on the number of trees}

The tree cover constructions of \cite{ADM+95} and \cite{BFN22} achieve the same bound of $O(\eps^{-d} \cdot \log(1/\eps))$
on the number of trees, which is basically the packing bound $O(\eps^{-d})$.
The Euclidean construction of \cite{ADM+95} is significantly more complex than the construction of \cite{BFN22} that applies to the wider family of doubling metrics. 
Here we give a short overview of the simpler construction of \cite{BFN22};  then we describe our Euclidean construction, aiming to focus on the geometric insights that we employed to breach the packing bound barrier.

The starting point of \cite{BFN22} is the standard \emph{hierarchy of $2^w$-nets} \{$N_w$\} \cite{GKL03}, which induces a hierarchical  \emph{net-tree}.\footnote{The standard notation in the literature on doubling metrics, including \cite{BFN22}, uses index $i$ instead of $w$ to refer to levels or distance scales; however, this paper focuses on Euclidean constructions, and we view it instructive to use a different notation.} 
Each net $N_w$ is greedily partitioned into a collection of $\Theta(\frac{2^w}{\eps})$-\emph{sub-nets} $N_{w,t}$, which too are hierarchical. 
For a fixed level $w$, the number of sub-nets $\{N_{w,t}\}$ is bounded by the packing bound $O(\eps^{-d})$, and each of them is handled by a different tree via a straightforward clustering procedure.
Na\"ively this introduces a $\log\Phi$ factor to the number of trees, each corresponding to a level ($\Phi$ is the aspect ratio of the point set).
The key observation to remove the dependency on the aspect ratio is that 
two far apart levels are more or less \emph{independent},
and one can pretty much use the same collection of trees for both.
More precisely, the levels are  partitioned into $\ell \coloneqq \log(1/\eps)$ \emph{congruence classes} $I_0, I_1, \dots, I_{\ell-1}$, where $I_j \coloneqq \Set{w \mid w \equiv j \pmod \ell}$. 
Since distances across different levels of the same class $I_j$ differ by at least a factor of $1/\eps$, it follows that 
all sub-nets $\{N_{w,t}\}_{w \in I_j}$
can be handled by a single tree
via a greedy hierarchical clustering.
Now the total number of trees is the number of sub-nets in one level, which is $O(\e^{-d})$, times the number of congruence classes $\log(1/\e)$. 

\medskip
Taking a bird's eye view of the construction of \cite{BFN22},
the following two-step strategy is used to handle pairwise distances within each congruence class $I_j$: 

\begin{enumerate}
\item \emph{Reduce the problem from the entire congruence class $I_j$ to a single level $w \in I_j$.} 
This is done by a simple greedy procedure. 
 
\item \emph{Handle each level $w \in I_j$ separately.}
This is done by a simple greedy clustering to 
the sub-nets $\{N_{w,t}\}$. 
\end{enumerate}

In Euclidean spaces, we shall use \emph{quadtree} which is the natural analog of the hierarchical net-tree. 
We too employ the trick of partitioning all levels in the hierarchy to congruence classes  \cite{CGMZ16,BFN22,LS22,ACRX22} and handle each one separately, and follow the above two-step strategy.  
However, 
the way we handle each of these two steps deviates significantly from \cite{BFN22}.

\paragraph{Step 1: Reduce the problem to a single level.}
At any level $w$, we handle every quadtree cell of width $2^w$ separately. 
Every cell is partitioned into subcells from level $w - \ell$ of width $\eps \cdot 2^w$, and each non-empty cell contains a single representative assigned by the construction at level $w - \ell$.  At level $w$, we construct a \emph{partial $(1+\eps)$-tree cover}, which roughly speaking  only preserves distances between pairs of representatives that are at distance roughly $2^w$ from each other; this is made more precise in the description of Step~2 below. 
Let $\tau(\eps)$ be the number of trees required for such a partial tree cover. 
To obtain a tree cover for all points in the current level-$w$ cell, we simply merge the aforementioned partial tree cover constructed for the level-$(w-\ell)$ representatives with the tree cover obtained previously for the points in the subcells. 
Finally, we choose one of those level-$(w-\ell)$ representatives as the level-$w$ representative of the current cell, and proceed to level $w + \ell$ of the construction. 

To achieve the required stretch bound, it is sufficient to guarantee that for every pair of points $(p,q)$,
some quadtree cell of side-length proportional to $\norm{pq}$ would contain both $p$ and $q$. 
Alas, this is impossible to achieve with a single quadtree. 
To overcome this obstacle, we use a result by Chan~\cite{Cha98}: there exists a collection of $\Theta(d)$ carefully chosen \emph{shifts} of the input point set, such that in at least one shift there is a quadtree cell of side-length at most $\Theta(d) \cdot \norm{pq}$
that contains both $p$ and $q$. 
The number of trees in the cover grows by a factor of $O_d(1)$.
Consequently, if each cell can be handled using $\tau(\eps)$ trees, then ranging over all the $\log(1/\eps)$ congruence classes and all the shifts, the resulting tree cover consists of $\tau(\eps) \cdot \log(1/\eps) \cdot O_d(1)$ trees; see Lemma~\ref{L:AddToMult} for a more precise statement.
The full details of the reduction are in Section~\ref{sec:reduction}.

\paragraph{Step 2: Handling a single level.}

Handling a single level is arguably the more interesting step, since this is where we depart from the general packing bound argument that applies to doubling metrics, and instead employ a more fine-grained geometric argument.
We next give a high-level description of the tree cover construction for a single level $w$. 
For brevity, in this discussion we focus on the 2-dimensional construction that does not use Steiner points. 
The full details, as well as  generalization for higher dimension 
and the Steiner tree cover construction, are given in Sections~\ref{sec:NonSteiner} and~\ref{sec:Steiner}.

We consider a single 2-dimensional quadtree cell of side-length $\Delta \coloneqq 2^w$ at level $w$, 
which is subdivided into subcells of side-length $\eps \cdot 2^w$.
Every level-$(w - \ell)$ cell has a representative and our goal is to construct a partial tree cover for any pair of representatives that are at a distance between $\Delta / 10$ and $\Delta$. 
(The final constants are slightly different; here we choose 10 for simplicity.)
To this end, we select a collection of  $\Theta(1/\e)$ \emph{directions}. 
For each direction $\nu$, we partition the plane into \emph{strips} of width $\e\Delta$, each strip parallel to $\nu$. We then shift each such partition orthogonally by $\e\Delta/2$; we end up with a collection of $2\cdot\Theta(1/\e)$ partitions, two for each direction. 
We call these partitions the \emph{major strip partitions}. 
Observe that for every pair of representative points $p$ and $q$, there is at least one major strip partition in some direction, such that both $p$ and $q$ are contained in the same strip. 
Crucially, we show that for every strip $S$ in a partition $P$, there is a collection of $O(1)$ trees that preserves distances between all points $p$ and $q$ in strip $S$ that are at distance between $\Delta/100$ and $\Delta$.
The key observation is that, since the strips in the same partition $P$ are \emph{disjoint} by design, the $O(1)$-many trees for each strip of $P$ can be combined into $O(1)$ forests. 
Thus the total number of forests is $O(1/\e)$.

To construct a collection of trees preserving distances within a single strip $S$, we first subdivide the strip $S$. If $S$ is in direction $\nu$, we partition $S$ into sub-strips \emph{orthogonal to $\nu$}, each of width $\Delta/20$. 
We call this a \emph{minor strip partition}. Observe that if points $p$ and $q$ are at distance $\ge \Delta/10$, they are in different sub-strips of the minor strip partition. 
For every pair of sub-strips $S_1$ and $S_2$ in the minor strip partition, we construct a single tree that preserves distances between points in $S_1$ and $S_2$ 
to within a factor of $1+\eps$. 
There are $O(1)$ sub-strips in the minor strip partition, so overall only $O(1)$ trees are needed for any strip~$S$.

\subsubsection{Bounding the degree}
The tree cover construction described above achieves the optimal bound on the number of trees, but the degree of points could be arbitrarily large.
While the previous tree cover constructions \cite{ADM+95,BFN22}
incur unbounded degree, the Euclidean construction of \cite{ADM+95}, when restricted to a single level in the hierarchy, achieves an absolute constant degree.%
\footnote{Although in the original paper of \cite{ADM+95} (as well as in \cite{NS07}) the bound is not an absolute constant, it was shown in \cite{Smid12} that an absolute constant bound can be obtained. 
Nonetheless, overlaying all levels of the hierarchy leads to a final degree bound of $\Theta(\log\Phi)$.}

In our construction, when restricted to a single level, the degree of points can be easily bounded by $O(1/\eps^2)$. 
However, 
in contrast to \cite{ADM+95}, our goal is to achieve this bound for the entire tree, \emph{across all levels} of the hierarchy. 
In particular, if we achieve this goal, the total degree of each point over all trees will be $O(\eps^{-1} \cdot \log(1/\eps)))$ ($O(\eps^{-d+1} \cdot \log(1/\eps))$ in general), which is optimal (up to logarithmic factor) due to the aforementioned lower bound~\cite{BT22}.
%
To achieve this goal, we strengthen the aforementioned two-step strategy as follows.
%
\paragraph{Step~1.}
In the reduction from the entire congruence class $I_j$ to a single level $w \in I_j$, the challenge is not to overload the same representative point over and over again across different levels of $I_j$.
To this end, we refine a \emph{degree reduction technique}, originally introduced by Chan \etal~\cite{CGMZ16} to achieve a bounded degree for $(1+\eps)$-stretch net-tree spanners in arbitrary doubling metrics.
The technique of \cite{CGMZ16} is applied on a \emph{bounded-arboricity} net-tree spanner, first by orienting its edges to achieve bounded out-degree for all points. 
Then, 
apply a greedy \emph{edge-replacement} process, where the edges are scanned in nondecreasing order of their level (or weight), and any incoming edge $(u,v)$ leading to a high-in-degree point $v$ is replaced by an edge leading to an incoming neighbor $w$ of $v$ in a sufficiently lower level, with $\norm{wv} \le \eps \norm{uv}$. 
It is shown that this process terminates with a bounded-degree spanner, where the degree bound is quadratic in the out-degree bound (arboricity) of the original spanner,
and the stretch bound increases only by an additive factor of $O(\eps)$.

We would like to apply this technique on every tree in the tree cover \emph{separately}; 
if instead we were to apply it on the union of the trees, we would create cycles; resolving them blows up the number of trees in the cover.
We demonstrate that by working on each tree separately,  not only does 
the greedy edge-replacement process 
reduce the degree in each tree to an absolute constant, but it also keeps the tree cycle-free as well as provides the required stretch bound;
see Section~\ref{sec:degreeGlobal}
for the details. 
In fact, it turns out to be advantageous to operate on each tree separately rather than on their union, since this way the out-degree bound in a single tree reduces to 1, which directly improves the total degree bound \emph{over all trees} to be linearly depending on $1/\eps$ 
rather than quadratically.
This is the key to achieving an optimal degree bound both within each tree as well as over all trees.

\paragraph{Step~2.~}
When handling a single level individually, the degree of points can be easily bounded by $O(1/\eps^2)$ as mentioned. 
However, we would like to achieve an absolute constant bound at each level.
Recall that, for every pair of sub-strips $S_1$ and $S_2$ in the minor strip partition of some strip $S$, we construct a single tree that preserves distances between points in $S_1$ and $S_2$ to within a factor of $1+\eps$; this tree is in fact a \emph{star}.
Perhaps surprisingly, every such star can be transformed into a \emph{binary tree} via a simple greedy procedure, with the stretch bound increased by just a factor of $1 + O(\eps \log(1/\eps))$; see Section~\ref{sec:degreeLocal} for the details.


\subsection{Organization}

In \Cref{sec:TreeCover}, we present the construction of tree covers in $\R^d$ with an optimal number of trees in both non-Steiner and Steiner settings, proving \Cref{thm:NonSteiner} and \Cref{thm:Steiner} for the plane. In \Cref{sec:highdim}, we generalize these constructions to $\R^d$ for arbitrary constant $d$. In
\Cref{sec:ConstDeg},
we reduce the degree of every tree in  the (non-Steiner) tree cover an absolute constant.
In \Cref{sec:route},
we show some applications of our tree cover to routing, proving \Cref{thm:routing}. 

\section{Optimal Tree Covers for Euclidean Spaces}\label{sec:TreeCover}

\subsection{Reduction to Partial Tree Cover}\label{sec:reduction}

Let $X$ be a set of points in $\mathbb{R}^d$. For any two points $p$ and $q$ in $X$, we use $\norm{pq}$ to denote their Euclidean distance. 
Without loss of generality we assume the minimum distance between any two points in $X$ is~$1$.

\begin{lemma}[Cf. \cite{Cha98,GH23}]
\label{lem:quadtree}
Let $L > 0$ be an arbitrary real parameter. 
Consider any two points $p,q \in [0,L)^d$, and let $\mathcal{T}$ be the infinite quadtree of\, $[0,2L)^d$. 
For $D \coloneqq 2\ceil{d/2}$ and $i = 0, \ldots, D$, 
let $\nu_i \coloneqq (iL/(D+1),\ldots,iL/(D+1))$. Then there exists an index $i \in \{0, \ldots,D\}$, such that $p + \nu_i$ and $q + \nu_i$ are contained in a cell of $\mathcal{T}$ with side-length at most $(4\ceil{d/2}+2) \cdot \norm{pq}$.
\end{lemma}

\begin{definition}
We call two points \EMPH{$(\mu,\Delta)$-far} if their distance is in $[\Delta/\mu,\Delta]$. 
\end{definition}

\begin{definition}
A \EMPH{$(\mu, \Delta)$-partial tree cover} for $X \subset \mathbb{R}^d$ with stretch $(1+\eps)$ is a tree cover with the following property: for every two $(\mu, \Delta)$-far points $p$ and $q$, there is a tree $T$ in the cover such that 
$\delta_T(p,q) \le (1+\eps)\cdot \norm{pq}$.
\end{definition}

\begin{lemma}[Reduction to partial tree cover]
\label{L:AddToMult}
Let $X$ be a set of points in $\mathbb{R}^d$, and let $\e$ be a number in $(0,1/20)$.
Suppose that for every $\mu > 0$, every set of points in $\R^d$ with diameter $\Delta$ admits a $(\mu,\Delta)$-partial tree cover with stretch $(1+\e)$, size $\tau(\e,\mu)$ and diameter of each tree at most $\gamma \Delta$ for some $\EMPH{$\gamma$} \ge 1$.
Then $X$ admits a tree cover with stretch $(1+\e)$ and size $O(d\cdot\log\frac{\gamma \cdot d\sqrt{d}}{\eps}\cdot \tau(\eps,\mu))$ with $\mu \coloneqq 10 d\sqrt{d}$.
\end{lemma}

\begin{proof}
Assume without loss of generality that the smallest coordinate of a point in $X$ is 0 and let $L$ be the largest coordinate in $X$. Let $\EMPH{$D$} \coloneqq 2\ceil{d/2}$ and let $\mathcal{Q}$ be the quadtree as in Lemma~\ref{lem:quadtree}. 
For $i \in \set{0,\ldots,D}$, let $\mathcal{Q}_i$ be $\mathcal{Q}$ shifted by $-\nu_i = (-iL/(D+1)),\ldots,-iL/(D+1))$. 

\medskip
\noindent\emph{Constructing the tree cover.} 
Let $\EMPH{$\ell$} \coloneqq \log\frac{\gamma d\sqrt{d}}{\eps}$
and let $\EMPH{$\mu$} \coloneqq 10d \sqrt{d}$. (Assume for simplicity that $\ell$ is an integer.)
Fix some $i \in \{0, \ldots,D\}$, $j \in \{0, \ldots, \ell-1\}$, and $k \in \{1, \ldots, \tau(\eps,\mu)\}$. 
We proceed to construct \EMPH{tree $T_{i,j,k}$}. 
Consider the \EMPH{congruence class} $\EMPH{$I_j$} \coloneqq \{z\ge 0 \mid z \equiv j \pmod \ell)\}$.  
The following construction is done for every $z \in I_j$ in increasing order.
Consider the level-$w$ quadtree $\mathcal{Q}_i$, with cells of width $2^w$.
If $w < \ell$, for each level-$w$ cell $C$, construct the $k$th among $\tau(\eps,\mu)$ trees from the $(\mu, 2^w)$-partial tree cover on the points in $C$,
and root it at an arbitrary point in $C$. 
For $w \ge \ell$, consider the subdivision of level-$w$ cell into subcells of level $w-\ell$. 
Let $X'$ be a subset of $X$ consisting of all the roots of the previously built subtrees in subcells of levels $w-\ell$.
Let $\EMPH{$\Delta_w$} \coloneqq 2^w \sqrt{d}$, and observe that $\Delta_w$ is an upper-bound on the diameter of $X'$.
Construct a $(\mu, \Delta_w)$-partial tree cover for $X'$ with $\tau(\eps,\mu)$ trees, and let $T$ be the $k$th tree of the $\tau(\e, \mu)$ trees constructed.
Take the previously built subtrees rooted at $X'$, and construct a new tree by identifying their roots with the vertices of $T$. 
Root this new tree arbitrarily. 
The tree $T_{i,j,k}$ is the final tree obtained after iterating over every $z \in I_j$. 

We prove the following two claims inductively.
\begin{claim}\label{clm:forestdiam}
Let $T_{i,j,k}^w$ be a tree constructed at level $w$ for $i \in \{0, \ldots,D\}$, $j \in \{0, \ldots, \ell-1\}$, $k \in \{1,\ldots,\tau(\eps,\mu)\}$ and $w \in I_j$.
\begin{enumerate}
\item $T_{i,j,k}^w$ is a tree.
\item $T_{i,j,k}^w$ has diameter $\phi_w$ at most $2\gamma \Delta_w$.
\end{enumerate}
\end{claim}

\begin{proof}
We prove the claim by induction over the level $w \in I_j$. 
\begin{enumerate}
\item The base case holds because the graph $T_{i,j,k}$ is initialized as a tree. For the induction step, consider some level $w \in I_j$ that is at least $\ell$. 
At this stage we construct a tree $T$ with vertex set consisting of representatives of the level $w - \ell$, and attach the trees rooted at each of the representatives we constructed previously to $T$. 
This is clearly a tree and the induction step holds. 
\item The base case holds because the diameter of each tree is at most $\gamma \Delta_w$, as guaranteed in the statement of \Cref{L:AddToMult}. For the induction step, we have $\gamma \Delta_w   + 2\gamma\Delta_{w-\ell} = \gamma \Delta_w + 2 \gamma \frac{\Delta_w \eps}{\gamma d \sqrt{d}} \le 2\gamma \Delta_w$. \qed
\end{enumerate}
\end{proof}

\begin{claim}\label{clm:count} 
The number of trees in the cover is $O(d \log\frac{\gamma d\sqrt{d}}{\eps} \cdot \tau(\eps,\mu))$.
\end{claim}

\begin{proof}
The trees $T_{i,j,k}$ are ranging over $(D+1) \cdot \ell \cdot \tau(\eps,\mu) =  (2\ceil{d/2}+1) \cdot \log\frac{\gamma d\sqrt{d}}{\eps} \cdot \tau(\eps,\mu)$ indices.
\end{proof}

\begin{claim}\label{clm:stretch}
For every two points $p,q\in X$, there is a tree $T$ in the cover such that $\delta_T(p,q) \le (1+\eps)\cdot \norm{pq}$, where $\delta_T(p,q)$ is the distance between $p$ and $q$ in $T$.
\end{claim}

\begin{proof}
By Lemma~\ref{lem:quadtree}, there exists a cell $C$ in one of the $D+1$ quadtrees which contains both $p$ and $q$ and has side-length $2^w \le (4\ceil{d/2}+2)\cdot \norm{pq} \le 5d \cdot \norm{pq}$. 
Let $\mathcal{Q}_i$ be such a quadtree, where $0 \le i \le D$, and let $0\le j \le \ell-1$ be such that $j \equiv w \pmod \ell$. 
Observe that $p$ and $q$ are $(\mu, \Delta_w)$-far. 
If $w < \ell$, we constructed a $(\mu, \Delta_w)$-partial tree cover of $C$, so the claim holds.
Otherwise suppose $w \ge \ell$. 
Recall that in the construction of the tree cover, we considered a subdivision of a level-$w$ cell (of side-length $2^w$) into smaller subcells of level $w-\ell$. 
For each subcell we choose a representative and constructed a tree cover on top of them. 
Let $p'$ (resp.\ $q'$) denote the representative of $p$ (resp.\ $q$) in the subcell at level $w - \ell$. We claim that $p'$ and $q'$ are $(\mu, \Delta_w)$-far, where $\Delta_w = 2^w \sqrt{d}$ denotes the diameter of the cell at level $w$. 
The bound $\norm{p'q'} \le \Delta_w$ follows from the fact that $p'$ and $q'$ are both in cell $C$. The distance can be lower-bounded as follows.
\begin{align*}
    \norm{p'q'} &\ge \norm{pq} - 2 \Delta_{w - \ell}
    \ge \frac{2^w}{5d} - 2 \cdot\frac{\e \cdot 2^w}{\gamma d \sqrt{d}} \sqrt{d}\\
    &= 2^w \left(\frac{1}{5d} - \frac{2 \e}{\gamma d}\right)
    \ge \frac{2^w}{10d}
    = \frac{\Delta_w}{\mu} &\text{as $\e < \frac{1}{20}$, $\gamma \ge 1$, and $\mu = 10d \sqrt{d}$}
\end{align*}

\noindent In other words, the representatives $p'$ and $q'$ are $(\mu, \Delta_w)$-far, meaning that one of the $\tau(\eps,\mu)$ trees $T$ in the partial tree cover for cell $C$ will preserve the stretch between $p'$ and $q'$ up to a factor of $(1+\eps)$. 
The distance between $p$ and $q$ in this tree can be upper bounded as follows.
\begin{align*}
\delta_T(p,q)  &\le  \delta_T(p,p') + \delta_T(p',q') + \delta_T(q,q')\\
&\le (1+\eps) \cdot \norm{p'q'} + \delta_T(p,p') + \delta_T(q,q')\\
&\le (1+\eps) \cdot (\norm{p'p} + \norm{pq} +\norm{qq'}) + \delta_T(p,p') + \delta_T(q,q') \\
&\le (1+\eps) \cdot (\norm{pq} +2\Delta_{w-\ell}) + 2\phi_{w-\ell} \\
&\le (1+\e) \cdot (\norm{pq} + 2 \Delta_{w-\ell} )+ 4\gamma \Delta_{w-\ell} &\text{by \Cref{clm:forestdiam}}\\
&\le (1+\eps) \cdot \left(\norm{pq} + 6 \Delta_{w} \cdot \frac{\eps}{ d\sqrt{d}}\right)\\
&= (1+O(\e)) \cdot \norm{pq}
\end{align*}
Stretch $1+\eps$ can be obtained by appropriate scaling.
\end{proof}
\Cref{clm:forestdiam,clm:count,clm:stretch} imply that the resulting construction is a tree cover with stretch $(1+\eps)$ and $O(d\log\frac{\gamma d\sqrt{d}}{\eps}\cdot \tau(\eps,\mu))$ trees, as required.
This concludes the proof of Lemma~\ref{L:AddToMult}. 
\end{proof}

\paragraph{Running Time.} 
Let \EMPH{$\mathrm{Time}_{\mu, \Delta}(n)$} be the time needed to construct a $(\mu, \Delta)$-partial tree cover for a given set of points of size $n$. 
In this paper, we assume that all algorithms are analyzed using the  real RAM model \cite{FW93, LPY05, SSG89, EVM22}. 
Constructing a (compressed) quadtree and computing the shifts require $O_d(n\log{n})$ time \cite{Cha98}. 
For each non-trivial node in the quadtree (a trivial node is a node that have only one child), we select a representative point, and then compute a $(\mu, \Delta)$-partial tree cover of the representative points corresponding to descendants of the node at $\ell = O(\log (1/\e))$ levels lower. 
Computing this $(\mu, \Delta)$-partial tree cover on $k$ representatives takes $\mathrm{Time}_{\mu, \Delta}(k)$ time. We can charge each of the $k$ representatives by $\mathrm{Time}_{\mu, \Delta}(k)/k$. Each of the $n$ points in our point set is charged $\ell = O_d(\log (1/\e))$ times. 
Assuming that $\mathrm{Time}_{\mu, \Delta}(a) + \mathrm{Time}_{\mu, \Delta}(b) \le \mathrm{Time}_{\mu, \Delta}(a +b)$, we can bound the total charge across all points by $O_d(\mathrm{Time}_{\mu, \Delta}(n) \cdot \log (1/\e))$.
Hence, the total time complexity is $O_d(n\log{n} + \mathrm{Time}_{\mu, \Delta}(n) \cdot \log (1/\e))$.

\subsection{Partial Tree Cover Without Steiner Points}
\label{sec:NonSteiner}

This part is devoted to the proof of Theorem~\ref{thm:NonSteiner}. 
We present the argument in~$\R^2$, and defer the proof for $\R^d$ with $d\geq 3$ to \Cref{subsec:non-Steiner-d3}.

\begin{lemma}\label{lem:NonSteiner2D}
Let $X$ be a set of points in $\mathbb{R}^2$ with diameter $\Delta$. For every constant $\mu >0$ there is a $(\mu, \Delta)$-partial tree cover for $X$ with stretch $(1+\eps)$ and size $O(1/\eps)$, where each tree has diameter at most $2\Delta\log(4\mu\eps)$.
\end{lemma}

The construction relies on partitioning the plane into strips.
Let $\theta$ be a unit vector. We define a \EMPH{strip in direction $\theta$} to be a region of the plane bounded by two lines, each parallel to $\theta$. 
The \EMPH{width} of the strip is the distance between its two bounding lines.
We define the \EMPH{strip partition with direction $\theta$ and width $w$} (shorthanded as \EMPH{$(\theta, w)$-strip partition}) to be the unique partition of $\R^2$ into strips of direction $\theta$ and width $w$, where there is one strip that has a bounding line intersecting the point $(0,0)$. 
Let \EMPH{$\theta^\bot$}$ \coloneqq (-\theta_y,\theta_x)$ be a unit vector perpendicular to $\theta$.
A $(\theta, w)$-strip partition with \EMPH{shift $s$} is obtained by shifting the boundary lines of the $(\theta, w)$ strip partition by $s \cdot \theta^\bot$.

Consider the following family of strip partitions: 
Let $\EMPH{$\theta_i$} \coloneqq (\cos(i\cdot \frac{\eps}{4\mu}), \sin(i\cdot \frac{\eps}{4\mu}))$ be the unit vector with angle $i\cdot \frac{\eps}{4\mu}$, for \smash{$i \in \set{0, \ldots, \frac{8\pi \mu}{\eps}-1}$}. 
Let set \EMPH{$\xi_i$} contains (1) the $(\theta_i, \e \frac{\Delta}{2\mu})$-strip partition with shift 0, and (2) the $(\theta_i, \e \frac{\Delta}{2\mu})$-strip partition with shift $\e \frac{\Delta}{4 \mu}$. 
Let $\EMPH{$\xi$} \coloneqq \bigcup_i \xi_i$.  
We call the strip partitions of $\xi$ the \EMPH{major strip partitions}. 
Clearly, $\xi$~contains $16 \pi \mu/ \e = O(1/\e)$ major strip partitions.
We define \EMPH{$\theta_i^\bot$} 
to be a vector orthogonal to $\theta_i$; and we define \EMPH{$\xi^\bot$} to be the set of all $(\theta_i^\bot, \frac{\Delta}{2\mu})$-strip partitions with shift 0, for every \smash{$i \in \set{0, \ldots, \frac{8\pi \mu}{\eps}-1}$}. 
We call the shift partitions of $\xi^\bot$ the \EMPH{minor strip partitions}. 
Every set~$\xi_i$ of major strip partitions is associated with a minor strip partition; notice that every major strip partition has an $\e$-factor smaller width to its orthogonal minor strip partition. 
See Figure~\ref{fig:strip-partitions}.

\begin{figure}[t]
    \centering
    \includegraphics[width=0.5\textwidth]{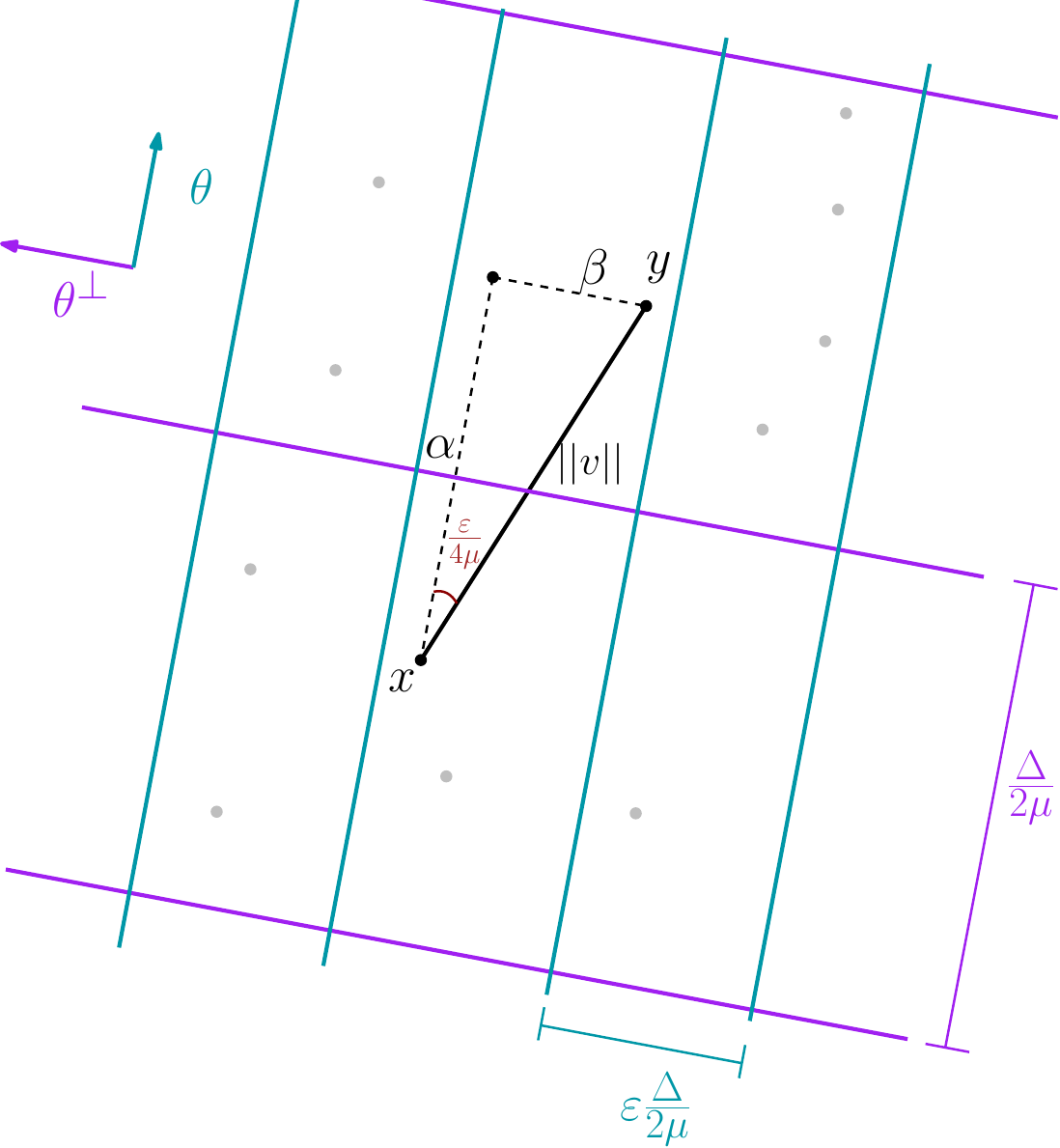}
    \caption{ A major strip partition (in blue) in direction $\theta$, and a minor strip partition (in purple) in direction~$\theta^\bot$. Points $x$ and $y$, and the vector $v$ broken into components parallel to and orthogonal to $\theta$. }
    \label{fig:strip-partitions}
\end{figure}

\begin{claim}
\label{clm:good-strip-nonsteiner}
    For any two points $x,y \in X$ such that $x$ and $y$ are $(\mu, \Delta)$-far, there exists some major strip partition $P \in \xi$ such that (1) the points $x$ and $y$ are in the same strip of $P$; and (2) in the associated minor strip partition $P^\bot \in \xi^\bot$, the points $x$ and $y$ are in different strips.
\end{claim}

\begin{proof}
Let $v$ denote the vector $y - x$. There exists some $i \in \set{0, \ldots, 8\pi\mu/\e -1}$ such that the angle between the vector $\theta_i$ and $v$ is at most $\e/4\mu$. We write $v$ as a linear combination of $\theta_i$ and a vector $\theta_i^\bot$ orthogonal to $\theta$:
$v = \alpha \cdot \theta_i + \beta \cdot \theta_i^\bot$.
As the angle between $v$ and $\theta_i$ is at most $\e/4\mu$ (and $\Delta/\mu \le \norm{xy} \le \Delta$), we have
\[
\abs{\alpha} \ge \norm{v}\cos \left(\frac\e{4\mu}\right) > \frac{\norm{v}}{2} \ge \frac{\Delta}{2\mu} \text{, and}
\]
\[
\abs{\beta} \le \norm{v} \sin \left(\frac{\e}{4\mu}\right) \le \frac{\e}{4\mu} \norm{v} \le \frac{\Delta}{4\mu}.
\]
Let $\xi_i$ be the set of major strip partitions in direction $\theta_i$.  As $\abs{\beta} \le \frac{\Delta}{4\mu}$, and $\xi_i$ consists of shifted strip partitions of width $\frac{\Delta}{2\mu}$, there is some major strip partition $P \in \xi_i$ in which $x$ and $y$ are in the same strip.
Further, every strip in the associated minor strip partition $P^\bot$ has width $\frac{\Delta}{2\mu}$, so the fact that $\abs{\alpha} > \frac{\Delta}{2\mu}$ implies that $x$ and $y$ are in different strips of $P^\bot$. This proves the claim.
\end{proof}

For every major strip partition in $\xi$, we now construct a tree which preserves approximately distances between points that lie in the same major strip but different minor strips. 
The following is the key claim.

\begin{claim}
\label{clm:nonsteiner-strip}
Let $S$ be a strip from a major strip partition in $\xi$, with direction $\theta$. Let $S_1$ and $S_2$ be two strips from a minor strip partition in $\xi$, both with direction $\theta^\bot$. Then there is a tree $T$ on $X \cap S$ such that for every $a \in X \cap S_1 \cap S$ and $b \in X \cap S_2 \cap S$,
\(
\norm{ab} \le \dist_T(a, b) \le \norm{ab} + \frac{\e \Delta}{\mu}.
\)
In particular, if $x$ and $y$ are $(\mu, \Delta)$-far, then $\norm{ab} \le \dist_T(a, b) \le (1+\e) \cdot \norm{ab}$.
\end{claim}

\begin{proof}
For any point $x \in \R^2$, we define \EMPH{$\score_\theta(x)$} to be the inner product $\langle x, \theta\rangle$.
Let $A \coloneqq X \cap S_1 \cap S$ and $B \coloneqq X \cap S_2 \cap S$.
As $A$ and $B$ belong to different minor strips in direction $\theta^\bot$,  without loss of generality $\score_\theta(a) < \score_\theta(b)$ for every $a \in A$ and $b \in B$.
Let $\EMPH{$a^*$} \coloneqq \arg\max_{a \in A} \score_\theta(a)$. 
We claim that for any $a \in A$ and $b \in B$,
\begin{equation}
\label{eq:nonsteiner-stretch}
    \norm{aa^*} + \norm{a^*b} \le \norm{ab} + \frac{\e \Delta}{\mu}.
\end{equation}
To show this, consider the line segment $\ell$ between $a$ and $b$. 
Let \EMPH{$L$} be the line in direction $\theta^\bot$ that passes through $a^*$. Because $\score_\theta(a) \le \score_\theta(a^*) \le \score_\theta(b)$, line $L$ and segment~$\ell$ intersect at some point \EMPH{$a'$} in the slab $S$; see Figure~\ref{fig:nonsteiner-stretch}. (Note that $a'$ is \emph{not} the projection of $a^*$ onto $\ell$.)
The distance $\norm{a^* a'}$ can be no greater than the width of the slab, so $\norm{a^* a'} \le \e \frac{\Delta}{2\mu}$. By triangle inequality, we have
\begin{align*}
    \norm{a a^*} + \norm{a^* b} &\le (\norm{aa'} + \norm{a' a^*}) + (\norm{a^* a'} + \norm{a' b})\\
    &\le \norm{aa'} + \norm{a'b} + \e \frac{\Delta}{\mu}\\
    &\le \norm{ab} + \frac{\e \Delta}{\mu}. 
\end{align*}
Let $T$ be the star centered at $a^*$, with an edge to every other point $x \in A \cup B$; the weight of the edge between $a^*$ and $x$ is $\norm{a^*x}$. For any $a\in A$ and $b \in B$, we clearly have $\norm{ab} \le \dist_T(a,b)$, and Equation~\eqref{eq:nonsteiner-stretch} guarantees that $\dist_T(a,b) \le \norm{ab} + \frac{\e \Delta}{\mu}$.
\end{proof}

\begin{figure}[ht!]
    \centering
    \includegraphics[width=0.5\textwidth]{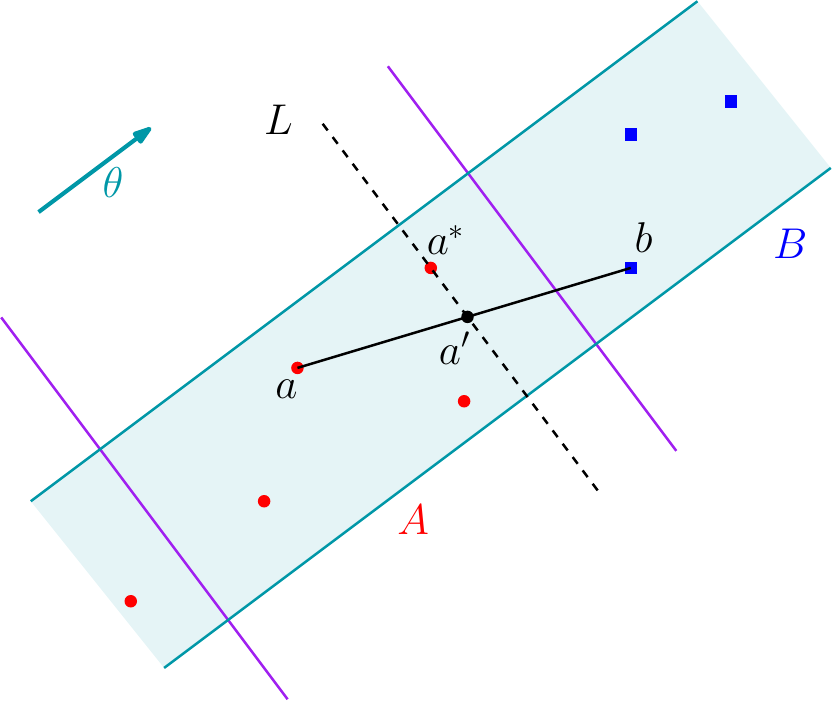}
    \caption{ Point sets $A$ and $B$, both in the same major strip (blue) but in different minor strips (purple). The points $a, a^*$, and $b$, with $\score_\theta(a) \le \score_\theta(a^*) \le \score_\theta(b)$, and the line $L$ passing through $a^*$.}
    \label{fig:nonsteiner-stretch}
\end{figure}

We can now prove Lemma~\ref{lem:NonSteiner2D}.
\begin{proof}[of 
Lemma~\ref{lem:NonSteiner2D}]
    Let $\xi$ be the set of major strip partitions defined above.
    Let $P$ be an arbitrary major strip partition in $\xi$, and let $P^\bot$ be the associated minor strip partition in $\xi^\bot$. 
    For each pair of strips $S_1$ and $S_2$ in $P^\bot$, we define \EMPH{tree $T_{P,S_1, S_2}$} as follows: 
    For every strip $S$ in $P$, apply Claim~\ref{clm:nonsteiner-strip} to construct a tree \EMPH{$T_S$} on (a subset of) $X \cap S$ that preserves distances between $X \cap S_1$ and $X \cap S_2$; and let \EMPH{$T_{P, S_1, S_2}$} be the tree obtained by joining together the trees $T_S$ from all strips $S$ in $P$. To join the trees, we build a balanced binary tree from the roots of $T_S$ for all strips $S$ in $P$.
    The \EMPH{tree cover $\mathcal{T}$} consists of the set of all trees $T_{P, S_1, S_2}$, for every major strip partition $P \in \xi$ and every pair of strips $S_1, S_2$ in the associated minor strip partition $P^\bot$.
    
    To bound the size of $\mathcal{T}$, observe that (1) there are at most $\frac{8\pi\mu}{\e} \cdot 2 = O(1/\e)$ major strip partitions containing points in $X$, and (2) for every strip $S$ in a major strip partition, at most $2\mu + 1 = O(1)$ strips in the associated minor strip partition contain points in $X \cap S$ (recall that point set $X$ has diameter $\Delta$).  Thus $\mathcal{T}$ contains $\frac{16\pi\mu}{\e} \cdot \binom{2\mu + 1}{2} = O(1/\e)$ trees.
    
    To bound the stretch, let $a$ and $b$ be arbitrary points in $X$. By Claim~\ref{clm:good-strip-nonsteiner}, there exists some major strip partition $P \in \xi$ such that (1) $a$ and $b$ are in the same strip in $P$; and (2) $a$ and $b$ are in different strips $S_1$ and $S_2$ of the associated minor strip partition $P^\bot$. Thus Claim~\ref{clm:nonsteiner-strip} implies that tree $T_{P, S_1, S_2}$ satisfies $\norm{ab} \le \dist_T(a,b) \le (1+\e)\cdot \norm{ab}$.

To bound the diameter, let $P$ be a major strip partition and let $S$ be a major strip in $P$. Observe that $T_S$ is a star and the distance from the root of $T_S$ to any other point in $T_S$ is at most $\Delta$. The roots of trees corresponding to strips in $P$ are connected by a binary tree by construction. Each edge of this binary tree is of length at most $\Delta$. The number of strips in $P$ is upper bounded by $2\mu/\eps$. Hence, the height of the binary tree is at most $\log(2\mu/\eps)$. This means that the diameter of the resulting tree is at most $2\cdot (\Delta + \log(2\mu/\eps)\cdot \Delta) = 2\Delta \log(4\mu/\eps)$.
\end{proof}

\paragraph{Running Time.} 
The inner product between each point with each vector $\theta_i$ can be precomputed using $O(|X| \cdot \frac{4\mu}{\eps})$ operations. 
For a major strip $S$, finding the maximum point in the intersection between $S$ and each of its minor strip only need time proportional to the number of points in $S \cap X$. 
Those points are chosen as roots of the stars corresponding to $S$. 
For each root, constructing the corresponding star requires $O(|S \cap X|)$ time. 
There are $\binom{2\mu + 1}{2}$ roots for each major strip.
Hence, the total time complexity of constructing the $(\mu, \Delta)$-partitial tree cover is:
\begin{equation*}
    |X| \cdot \frac{4\mu}{\eps} + \binom{2\mu + 1}{2}\sum_{\text{major strip $S$}} |S \cap X| = O(|X|\cdot \eps^{-1})
\end{equation*}
Therefore, the time complexity of constructing the tree cover is $O_d(n\log{n} + n\eps^{-1}\log(1/\eps))$.

\subsection{Partial tree cover with Steiner points}
\label{sec:Steiner}

This part is devoted to the proof of \Cref{thm:Steiner} for $\R^2$; the argument for dimension $d\geq 3$ is deferred to \Cref{subsec:Steiner-d3}.

\begin{lemma}\label{Steiner2D}
Let $X$ be a set of points in $\mathbb{R}^2$ with diameter $\Delta$. 
For every constant $\mu$, there is a Steiner $(\mu, \Delta)$-partial tree cover with stretch $(1+\eps)$ for $X$ with $1/\sqrt{\eps}$ trees, where each tree has diameter at most~$3\Delta$.
\end{lemma}

\noindent Consider a square of side-length $\Delta$ containing $X$, and let $\mu$ be an arbitrary constant. 
Divide the square into vertical slabs of width \smash{$\frac{\Delta}{3\sqrt{2}\mu}$} and height $\Delta$, and into horizontal slabs of width $\Delta$ and height \smash{$\frac{\Delta}{3\sqrt{2}\mu}$}.

\begin{observation}
\label{obs:dif_side_slab}
For any two points $p,q \in X$ such that $p$ and $q$ are $(\mu,\Delta)$-far, there exists either a horizontal or a vertical slab such that $p$ and $q$ are from different sides of the slab.
\end{observation}

\begin{proof}
Suppose towards contradiction 
there are two adjacent horizontal slabs containing both $p$ and $q$ and also two adjacent vertical slabs containing both $p$ and $q$. 
The distance between $p$ and $q$ is at most $\norm{pq}\le 2\cdot \frac{\Delta}{3\sqrt{2}\mu}\cdot \sqrt{2} < \frac{\Delta}{\mu}$, contradicting the assumption that $p$ and $q$ are $(\mu,\Delta)$-far.
\end{proof}

For each horizontal (resp.\ vertical) slab $S$, we consider the horizontal (resp.\ vertical) line segment $\ell$ that cuts the slab into two equal-area parts. The length of $\ell$ is $\Delta$. 
Let $k \coloneqq \floor{2\mu/ \sqrt{\eps}}$ be an integer. 
We partition $\ell$ into $k$ intervals, called $[a_0, a_1], [a_2, a_3], \ldots [a_{k - 1}, a_{k}]$, each of length $\sqrt{\eps}\Delta / 2\mu$. 
For each point $a_i$,
we construct \EMPH{tree $T^i_S$} by adding edges between $a_i$ and every point in $X$. Finally, connect the points $a_i$ using a straight line and let $T$ be the resulting tree. The diameter of this tree is at most $3\Delta$.

\begin{claim}
\label{clm:steiner-stretch}
    For any two points $p, q \in X$ such that $p$ and $q$ are $(\mu, \Delta)$-far, there exists a slab $s$ and an integer $i \in \set{0, \ldots, k}$ such that $\delta_{T^i_S}(p, q) \leq (1 + \eps) \cdot \norm{pq}$.
\end{claim}

\begin{proof}
    By Observation \ref{obs:dif_side_slab}, there exists a slab $S$ such that $p$ and $q$ are in different sides of it. 
    Without loss of generality assume that $S$ is horizontal. 
    By construction, we partition the middle interval $\ell$ of $S$ into $k$ intervals $[a_0, a_1], [a_2, a_3], \ldots [a_{k - 1}, a_{k}]$ each of length $\sqrt{\eps}\Delta / 2\mu$. 
    Let $r$ be the intersection between $pq$ and $\ell$, and let $a_i$ be the closest point to $r$.
    Let $r'$ be the projection of $a_i$ to $r'$. Hence, $||a_ir'|| \leq ||a_ir|| \leq \sqrt{\eps}\Delta / 2\mu$. Using the triangle inequality, we have: 
    \begin{equation}
        \label{eq:bdd_dist}
        \begin{split}
            \delta_{T^i_S}(p, q) \leq ||pa_i|| + ||a_iq|| =  \sqrt{||pr'||^2 + ||r'a_i||^2} + \sqrt{||r'q||^2 + ||r'a_i||^2}.
        \end{split}
    \end{equation}
    Observe that $||pr'|| \geq \Delta/2\mu$. Thus, $||r'a_i|| \leq \sqrt{\eps}\Delta/2\mu \leq ||pr'|| \sqrt{\eps}$. Similarly, $||r'a_i|| \leq ||r'q|| \sqrt{\eps}$. Combining with Equation \ref{eq:bdd_dist}, we get:
    \begin{equation*}
        \begin{split}
            \delta_{T^i_S}(p, q) &\leq ||pa_i|| + ||a_iq|| =  \sqrt{||pr'||^2 + \eps||pr'||^2} + \sqrt{||r'q||^2 + \eps||r'q||^2}\\
            &\leq \sqrt{1 + \eps} \cdot (||pr'|| + ||r'q||) \leq (1 + \eps) \cdot ||pq||.
        \end{split}
    \end{equation*}
    \vspace{-4pt}
    \aftermath
\end{proof}

We now prove \Cref{Steiner2D}. Let $\mathcal{T}$ be the set containing trees $T^i_s$ for every horizontal or vertical slabs $s$ and every index $i \in [0,k]$.
There are $O(\mu) = O(1)$ horizontal and vertical slabs, so $\mathcal{T}$ contains $O(k) = O(1/\sqrt{\e})$ trees. It follows immediately from \Cref{clm:steiner-stretch} that $\mathcal{T}$ is a Steiner $(\mu, \Delta)$-partial tree cover for $X$ with stretch $(1+\e)$.

\paragraph{Running Time.} 
For a set $X$, creating the set of slabs can be done in $O(1)$ time. For each slab, finding a net of the middle line takes $O(1/\sqrt{\eps})$ time. 
For each Steiner point, it requires $O(|X|)$ time to create a tree connecting that point to everyone in $X$. 
Totally, the time complexity is $O(|X| / \sqrt{\eps})$. 
Therefore, the time complexity of constructing the tree cover is $O_d(n\log{n} + n\eps^{-1/2}\log(1/\eps))$.

\section{Tree Cover in Higher Dimensions}
\label{sec:highdim}

\subsection{Non-Steiner tree covers} \label{subsec:non-Steiner-d3}

We now prove an analog of Lemma~\ref{lem:NonSteiner2D} in $\R^d$, for any constant $d = O(1)$. The definition of strip partition and the sets $\xi$ and $\xi^\bot$, are different in $\R^d$ than in $\R^2$. Let $\theta$ be a vector. An \EMPH{$\R^d$-strip with direction $\theta$ and width $w$} is a convex region $S \subset \R^d$ such that there is a line $\ell$ in $S$ such that every point in the strip is within distance at most $w/2$ of $\ell$. 
The line $\ell$ is called the \EMPH{spine} of the strip. 
An $\R^d$-strip partition is a partition of $\R^d$ into $\R^d$-strips.
For the construction of the major strip partitions $\xi$, we use the following well-known lemma, slightly adapted from a version in the textbook by Narasimhan and Smid \cite{NS07}.
\begin{lemma}[Cf. Lemma~5.2.3 of \cite{NS07}]
\label{lem:direction-net}
    Let $\e$ be a number in $(0,1)$. There is a set $\mathcal{V}$ of vectors in $\R^d$ such that (1) $\mathcal{V}$ contains $O(\e^{1-d})$ vectors, and (2) for any vector $v$ in $\R^d$, there is some vector $v' \in \mathcal{V}$ such that the angle between $v$ and $v'$ is at most $\e$.
\end{lemma}
We also use a variant of the shifted quadtree construction of Chan \cite{Cha98} (which follows immediately from our Lemma~\ref{lem:quadtree}).

\begin{lemma}[Cf. \cite{Cha98}]
\label{lem:shifted-grid}
    For any constant $\Delta > 0$, there is a set $\mathcal{P}$ of partitions of $\R^d$ into hypercubes of side length $(4\ceil{d/2}+2)\Delta$ such that (1) there are $O(d)$ partitions in $\mathcal{P}$, and (2) for every pair of points $x,y\in \R^d$ with $\norm{xy} \le \Delta$, there is some partition $P \in \mathcal{P}$ where $x$ and $y$ are in the same hypercube in~$P$.
\end{lemma}

Let $\theta$ be a vector in $\R^d$. We define \EMPH{$X_\theta$} to be the hyperplane orthogonal to $\theta$. We can view $X_\theta$ as a copy of $\R^{d-1}$. Let $P_\theta$ be an arbitrary partition of $X_\theta$ into $\R^{d-1}$-hypercubes
with side length $\e \frac{\Delta}{2\mu d}$. This partition induces an $\R^d$-strip partition with direction $\theta$ and width $\e \frac{\Delta}{2\mu}$: for every hypercube $R$ in the partition $P_\theta$, the corresponding strip is defined by $\set{r + \alpha \cdot \theta : r \in R, \alpha \in \R}$. We denote this strip partition as \EMPH{$S(P_\theta)$}.
The fact that $S(P_\theta)$ has width $\e\frac{\Delta}{2\mu}$ follows from the fact that every point in a $\R^{d-1}$-hypercube of side length $\e \frac{\Delta}{2\mu d}$ is within distance $\e \frac{\Delta}{4\mu}$ of the center point of the hypercube.

We now define the set \EMPH{$\xi$} of major strip partitions. Let $\mathcal{V}$ be the set of vectors provided by Lemma~\ref{lem:direction-net}, setting the parameter $\e' = \frac{\e}{10 \mu d^2}$. For every $\theta \in \mathcal{V}$, let $\mathcal{P}_{\theta}$ denote the set of partitions of $X_{\theta}$ into $\R^{d-1}$-hypercubes of side length $\e \frac{\Delta}{2 \mu d}$, as guaranteed by Lemma~\ref{lem:shifted-grid}. The set \EMPH{$\xi_\theta$} contains the $(\theta, \e \frac{\Delta}{2\mu})$-strip partitions $S(P_\theta)$ associated with every $P_\theta \in \mathcal{P}_\theta$.
Define $\xi = \bigcup \xi_\theta$.
The following observation is immediate from Lemma~\ref{lem:shifted-grid}: 
\begin{observation}
\label{obs:shifted-strip}
Let $\theta$ be a vector in $\R^d$. If $x$ and $y$ are two points whose projections onto $X_\theta$ are within distance $\e \frac{\Delta}{10 \mu d^2}$, then there is some strip partition in $\xi_\theta$ with a strip containing both $x$ and $y$.
\end{observation}

We now define \EMPH{$\xi^\bot$}, the set of minor strip partitions.
For every $\theta \in \mathcal{V}$, let $\theta^\bot$ be some arbitrary vector that is orthogonal to $\theta$. Let $P_{\theta^\bot}$ be an arbitrary partition of $X_{\theta^\bot}$ into $\R^{d-1}$-hypercubes with side length $\frac{\Delta}{2 \mu d}$. Define $\xi^\bot$ to be the set containing the $(\theta^\bot, \frac{\Delta}{2\mu})$-strip partition $S(P_{\theta^\bot})$ for every $\theta \in \mathcal{V}$.
With these modified definition of $\xi$ and $\xi^\bot$, the claims from the $\R^2$ case generalize naturally.
\begin{itemize}
    \item The proof of Claim~\ref{clm:good-strip-nonsteiner} is similar to the $\R^2$ case. We break $v = y - x$ into a component parallel to $\theta$ and a component that lies in the hyperplane orthogonal to $\theta$; the former has length $\alpha > \frac{\Delta}{2\mu}$ and the latter has length $\beta \le \e \frac{\Delta}  {10\mu d^2}$. Observation~\ref{obs:shifted-strip} (together with the upper-bound on $\beta$) guarantees that there is some major strip partition in direction $\theta$ in which $x$ and $y$ are in the same strip.  The lower bound on $\alpha$ implies that $x$ and $y$ are in different strips of the associated minor strip partition in direction $\theta^\bot$.
    \item In the proof of Claim~\ref{clm:nonsteiner-strip}, the only difference is that the line $L$ in the 2D case is replaced by a hyperplane $L$ orthogonal to $\theta$. To show that $\norm{a^*a} \le \e \frac{\Delta}{2\mu}$, we argue as follows. Hyperplane $L$ intersects the spine of the strip at some point $s$; as the width of the strip is $\e \frac{\Delta}{2\mu}$, every point in $L$ that is in the strip (which includes $a^*$ and $a$) is within distance $\e \frac{\Delta}{4\mu}$ of $s$.
    Triangle inequality proves the claim, and the rest of the proof carries over.
    \item The proof of Lemma~\ref{lem:NonSteiner2D} carries over almost exactly. The size of $\xi$ is $O(\e^{1-d})$. For every major strip partition $P \in \xi$, there are $\binom{(2\mu d + 1)^{d-1}}{2} = O(1)$ pairs of strips in the corresponding minor strip partition of $\xi^\bot$, and thus the tree cover $\mathcal{T}$ contains $O(\e^{1-d})$ trees. The stretch bound carries over without modification.
\item The diameter of each of the trees is $\Delta \log \frac{8\mu d}{\eps}$. Every major strip partition is induced by set of $\mathbb{R}^{d-1}$-hypercubes with side length $\eps \frac{\Delta}{2\mu d}$. Since the diameter of the point set is $\Delta$, the number of hypercubes required is at most $\frac{2\mu d}{\eps}$. The same argument as in the 2-dimensional case implies the claimed bound on the diameter. The height of the binary tree is at most $\log\frac{2\mu d}{\eps}$ and each edge is of length at most $\Delta$. The diameter is at most $2(\Delta + \Delta \log\frac{2\mu d}{\eps} )= 2\Delta \log \frac{4\mu d}{\eps}$.
\end{itemize}
Together with the reduction to a fixed scale (Lemma~\ref{L:AddToMult}), this proves Theorem~\ref{thm:NonSteiner}.

\bigskip \noindent\textbf{Running Time: } The running time analysis is similar to the $2D$ case. The inner product between each point to each direction vector can be precomputed in $O(\eps^{1 - d})|X|$ time. For each major strip $S$, there are at most $\binom{(2\mu d+1)^{d-1}}{2}$ pair of minor strips that intersect $S \cap X$. Hence, the total running time to construct a $(\mu, \Delta)$-partial tree cover is:
\begin{equation*}
    \eps^{1-d}|X| + \binom{(2\mu d+1)^{d-1}}{2} \cdot \sum_{\text{$S$ is a major strip}}|S \cap X| = O_d( \eps^{1-d}|X|)
\end{equation*}
Then, the time complexity of constructing the tree cover is $O_d(n\log{n} + n\eps^{1 - d}\log{1/\eps})$.

\subsection{Tree covers with Steiner points}\label{subsec:Steiner-d3}

The construction for $d$-dimensional Euclidean space is a direct generalization of two-dimensional case. 
Consider a hypercube of side length $\Delta$. 
We divide the hypercube by each coordinate into slabs of height \smash{$\frac{\Delta}{3\sqrt{d}\mu}$}; all other sides have length $\Delta$. 
One can think of each slab as a $d$-dimensional rectangle, joined by two $(d-1)$-hypercubes that are distance \smash{$\frac{\Delta}{3\sqrt{d}
\mu}$} away from each other. 
Analog of Obs.~\ref{obs:dif_side_slab} follows.

\begin{observation}
    For any two points $p, q \in X$ such that $p$ and $q$ are $(\mu, \Delta)$-far, there exists a slab such that $p$ and $q$ are from different sides of it.
\end{observation}

Each tree is constructed similarly to the two-dimensional case. 
For each slab, we find a $\frac{\sqrt{\eps}\Delta}{2\mu}$-net for the $(d - 1)$-hypercube at the middle of each slab. 
For each net point $u$, we create a tree connecting $u$ to every vertex in $X$. The proof follows similarly. 
The total number of trees is 
\[
O\Paren{ d \cdot 3\sqrt{d}\mu \cdot \left(\frac{2\mu}{\sqrt{\eps}}\right)^{d - 1} } 
= 
O\Paren{d^{3/2}2^d\mu^{d}\eps^{(d - 1)/2} }.
\]
The diameter can be bounded by $3\Delta$, as in the 2-dimensional argument.

\bigskip \noindent\textbf{Running Time: } Creating the set of slabs requires time equal to the number of slabs, which is $d\cdot 3\sqrt{d}\mu = O(d^{3/2})$. For each slab, we find a $\frac{\sqrt{\eps} \Delta}{2\mu}$-net of a $(d - 1)$-hypecube, which has $(2\mu/\sqrt{\eps})^{d - 1}$ points. For each Steiner point, we connect it to every point in $X$ in $O(|X|)$ time. Hence, the total time complexity is:
    \begin{equation*}
        O(d^{3/2} \cdot (2\mu/\sqrt{\eps})^{d - 1}\cdot |X|) = O_d(\eps^{(1 - d)/2}|X|). 
    \end{equation*}
Then, the time complexity of constructing the tree cover is $O_d(n\log{n} + n\eps^{(1 - d)/2}\log{1/\eps})$.
\section{Constant degree constructions}\label{sec:ConstDeg}

In this section, we prove the following theorem.
\begin{theorem}\label{thm:boundedDegree}
For every set of points in $\mathbb{R}^d$ and any $0 < \eps < 1/16$, there exists a tree cover with stretch $1+\eps$ and $O_d(\eps^{-(d-1)} \log 1/\e)$ trees such that every \emph{metric point} has a bounded degree.
\end{theorem}

Our tree cover construction is a collection of trees, each of which possibly uses a copy of the same point many times. Each tree is constructed iteratively, going from smaller scales to the larger ones.
In Section~\ref{sec:degreeGlobal} we use the degree reduction technique due to \cite{CGMZ16} for each tree in the cover. This allows us to bound the degree in terms of the number of trees in the cover and the degree at a single scale in the construction. In Section~\ref{sec:degreeLocal}, we show that the degree at a single scale is constant.

\subsection{Bounding the degree of metric points}\label{sec:degreeGlobal}
Consider a single tree in the cover. Our tree cover construction from Section~\ref{sec:NonSteiner} does not use Steiner points, but it still might consider the same point from the metric $X$ across multiple levels of construction. Even if at each level of the construction every node has a bounded degree (which we show how to achieve in Section~\ref{sec:degreeLocal}) the degree of each metric point might still be unbounded. To remedy this, we apply the degree reduction technique of \cite{CGMZ16}.

Start from the tree cover construction from Section~\ref{sec:TreeCover} and fix a tree $T = (V,E)$ from the cover. Let $\ell =\log(d\sqrt{d}/\eps)$ be the same as in \Cref{sec:reduction}. Without loss of generality, assume the tree was constructed in the congruence class $I_j \coloneqq \Set{\big. z\mid z \equiv j \pmod \ell }$. Assume that at the every level of the construction, the edges of the tree are oriented from the parents to the children so that the outdegree of each node is $\alpha$ and indegree is $1$. We show how to bound the degree of every point of the metric with respect to $T$. 

Let \EMPH{$i^*(v)$} be the highest quadtree level at which point $v$ is considered as a representative. 
For every edge $(u,v)$ in $T$, we orient it from $u$ to $v$ if $i^*(u) < i^*(v)$. If $i^*(u) = i^*(v)$, break the ties according to the tree structure, from children towards parent. 
We use \EMPH{$\hat{E}$} to denote the set of arcs obtained in this way. 
Note that $|\hat{E}| = |E|$, since we do not change any edges. 
Next, we describe the modification of $\hat{E}$, where we replace some edges of $\hat{E}$ and obtain the set of \EMPH{$\tilde{E}$}. Let $u$ be a vertex at level $i$ and let $\hat{E}_i$ be the set of edges used in the tree constructed at level $i$. Let \EMPH{$M_i(u)$} be the set of endpoints of edges in $\hat{E}_i$ oriented into $u$.
Let $\EMPH{$\mathcal{I}_u$} \coloneqq \{ i \mid M_i(u) \neq\varnothing \}$. Suppose that the indices in $\mathcal{I}_u$ are ordered increasingly. 
Next we modify arcs going into $u$ as follows. 
Keep $M_{i_1}(u)$ directed into $u$. For $j > 1$ we pick an arbitrary vertex $w \in M_{i_j-\ell}(u)$ and for each point $v \in M_{i_j}(u)$ replace arc $(v,u)$ by an arc $(v,w)$. 

\begin{claim}
\label{clm:singleLevelOut}
If at every level $T$ has an outdegree $\alpha$, then every metric point has outdegree $\alpha$. Moreover, every node with an outgoing edge at level $i$ ceases to be considered at levels higher than $i$.
\end{claim}

\begin{proof}
Consider an arc $(u,v)$, i.e., an edge $(u,v)$ directed from $u$ towards $v$. This means that $i^*(u) \le i^*(v)$. Let $i$ be the level at which $(u,v)$ was added to $T$. Recall that the edges are added while handling a single quadtree cell at level $i$ and only one point from the cell is chosen a representative for the subsequent handling of level $i+\ell$. If $i^*(u) < i^*(v)$, this means that $v$ is a representative and $u$ does not exist on any subsequent level starting from $i+\ell$. 
If $i^*(u)=i^*(v)$, then neither $u$ nor $v$ are representatives, since the representative exists at a level higher than $i^*(u)$. Hence, $u$ does not exist on any subsequent level starting from $i+\ell$. 
In conclusion, $u$ can have outgoing edges only at a single level of construction.
\end{proof}

\begin{claim}
\label{clm:degreeGlobal}
If at every level $T$ has an outdegree $\alpha$ and indegree $\beta$, then every metric point has degree with respect to $T$ at most $\alpha+\beta+\alpha\beta$.
\end{claim}

\begin{proof}
Consider a metric point $w$. 
There are at most $\alpha$ edges directed out of $w$ by Claim~\ref{clm:singleLevelOut}. 
Out of the edges that were directed into $w$ in $\hat E$, there are only edges from $M_{i_1}(w)$ that remained directed into $w$. 
There are at most $\beta$ such edges. 
Finally, some new edges might have been attached to $w$ due to the modification into $\tilde{E}$. 
Consider an arc $(w,u)$ directed out of $w$; there is a unique level $i_j$ where $w$ is in $M_{i_j-\ell}(u)$.  
Only edges of the form $(v,u)$ at level $i_j$ can be redirected to $(v,w)$ by the modification process; each such $v$ must be an in-neighbor of $u$ at level $i_j$.
A counting argument shows that for each arc $(w,u)$ going out of $w$, 
there are at most $\beta$ new arcs attached to $w$, each attribute to an in-neighbor of $u$ in $\hat{E}$ before the modification; and there are $\alpha$ possible choices of $u$, all being the out-neighbors of $w$.
Putting everything together, the bound on the degree is $\alpha + \beta + \alpha\beta$.
\end{proof}

We next show that the modification of the edges of $T$ does not create cycles. 
\begin{claim}
The modified tree does not contain cycles.
\end{claim}
\begin{proof}
Suppose towards contradiction that the modified $T$ contains a cycle and let $(v,w)$ be the first edge in the modification process whose insertion caused a cycle. Recall that the arc $(v,w)$ gets inserted in place of arc $(w,u)$, where $v \in M_{i_j}(u)$ and $w \in M_{i_j-\ell}(u)$, for some levels  $i_j$ and $i_{j-\ell}$.  Since $(v,w)$ introduces a cycle, this means that $T$ contains an alternative path between $v$ and $W$.  By Claim~\ref{clm:singleLevelOut}, node $w$ does not exist at level higher than $i_{j-\ell}$. Hence, the path appears at some level lower than $i_{j-\ell}$. In the original tree $T$, there were edges $(v,u)$ and $(w,u)$. Together with the path between $w$ and $v$, this creates a cycle in the original $T$, a contradiction.
\end{proof}

Next, we show that the stretch does not increase by more than a $(1+O(\eps))$ factor.
\begin{claim}
Let $\tilde{d}_T$ be the metric induced by $\tilde{T}$. Then, for every $u, v \in V(T)$, it holds $\norm{uv} \le (1+O(\eps)) d_T(u,v)$.
\end{claim}
\begin{proof}
It suffices to show that for an arc $(v,u)$ that is removed from $\hat{E}$ it holds  $\tilde{d}_T(u,v) \le (1+O(\eps)) d_T(u,v)$. Let $s=\lfloor j/\ell \rfloor$.
By construction, since $(v,u)$ is removed, there exists points $v_0, v_1, \ldots, v_s$, such that $v = v_0$, $(v_s, u) \in \tilde{E}$ and for   $0\le k < s$:  $(v_k, v_{k+1})\in \tilde{E}$, and $v_k \in M_{j-k\ell}(u)$. Recall that we use $\Delta_w = 2^{w}\sqrt{d}$ to denote the diameter of a quadtree at level $w$.
\begin{observation}\label{obs:epsDecay}
For every $0 \le k < s$, $\norm{uv_{k+1}} \le \eps^\ell\norm{uv_k}$.
\end{observation}
Applying \Cref{obs:epsDecay} inductively, we can prove that $\norm{uv_{k}}\le \eps^{\ell k}\norm{uv_0}$ for every $0 \le k \le s$. We can also bound $\norm{v_kv_{k+1}} \le \norm{v_ku} + \norm{uv_{k+1}}\le (1+\eps)\norm{v_ku} \le (1+\eps) \eps^{\ell k}\norm{uv_0}$.
By triangle inequality $\norm{uv}$ can be upper bounded by the length of the path $\langle v_0, v_1, \ldots, v_s, u \rangle$.
\begin{align*}
\tilde{d}_T(u,v) &\le \norm{v_su}+\sum_{0 \le k < s}\norm{v_k v_{k+1}}\\
&\le \eps^{-\ell s}\norm{uv_0} + \sum_{0 \le k < s}(1+\eps)\eps^{\ell k}\norm{uv_0}\\
&\le (1+O(\eps)) \norm{uv_0}
\end{align*}
\aftermath
\end{proof}
In the next subsection, we show that $\alpha=1$ and $\beta=5$. Plugging in Claim~\ref{clm:degreeGlobal}, we obtain the bound of $\alpha+\beta+\alpha\beta=11$ on the degree of every node in the metric.

\subsection{Bounding the degree of tree nodes}\label{sec:degreeLocal}

Recall our construction from Section~\ref{sec:NonSteiner}. 
Consider a single strip  $S$ and the sets $A$ and $B$ as in the proof of \Cref{clm:nonsteiner-strip}. 
Let $\theta$ be the direction of strip $S$. 
The tree handling the distances between $A$ and $B$ is a star rooted at a point $\EMPH{$a^*$} \coloneqq \argmax_{a \in A}\score_\theta(a)$.  We describe three different constructions. The first construction achieves a constant degree but it requires scaling after which the number of trees grow by a factor of roughly $\log(1/\eps)^d$. The second construction achieves degree of roughly $2^d$ and does not require any scaling. The third construction achieves degree 5 and requires scaling so that the number of trees grows by a factor of roughly $d^{d-1} = O_d(1)$.

\paragraph{Constant degree, simple attempt.} 
Let $\EMPH{$A'$} \coloneqq \langle a^*, a_1, a_2, \ldots \rangle$ be the set of points in $A$, sorted in decreasing order  with respect to $\score_\theta$. 
We make a balanced binary tree $T_A$ rooted at $a_1$ such that for every node $a_i$ and its parent $a_j$, $\score_\theta(a_i) \le \score_\theta(a_j)$. To do so, we mark $a_1$ as \emph{visited} and make it the root of $T_A$. 
Next, we scan the points $a_2,a_3,\ldots$ in order, make $a_i$ child of the node in $T_A$ that was visited earliest and still has 0 or 1 children, then mark $a_i$ visited. 
We similarly construct $T_B$ satisfying that for every node $b_i$ and its parent $b_j$, $\score_\theta(b_i) \ge \score_\theta(b_j)$. Finally, let $T$ be the tree rooted at $a^*$ having subtrees $T_A$ and $T_B$ as its children.

\begin{figure}[h!]
    \centering
    \includegraphics[width=0.5\textwidth]{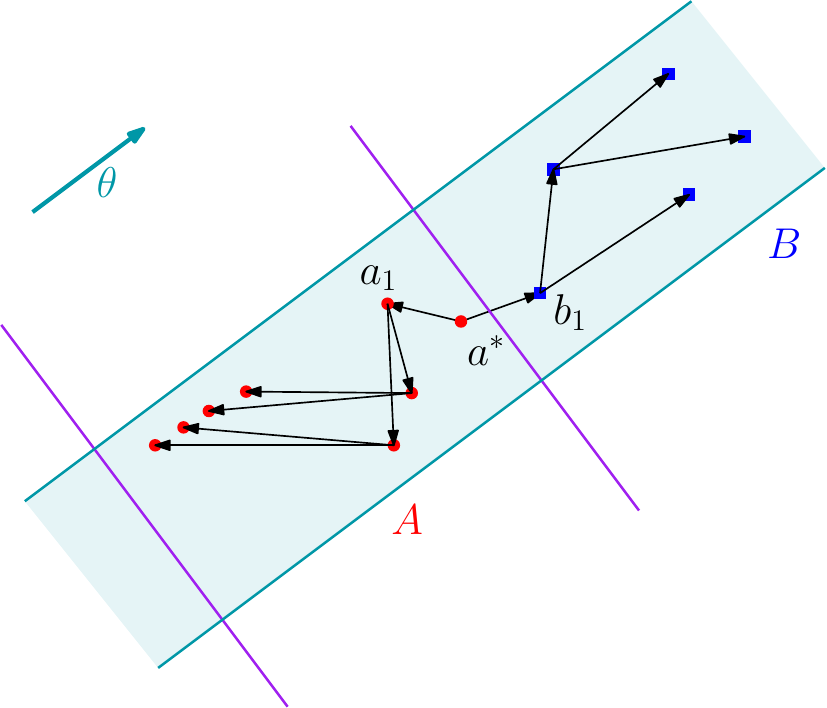}
    \caption{The binary trees $T_A$ and $T_B$, constructed greedily from point sets $A$ and $B$}
    \label{fig:enter-label}
\end{figure}

Recall that, during the reduction to a single scale (Lemma~\ref{L:AddToMult}), we only construct the partial tree covers of Section~\ref{sec:NonSteiner} on $1/\e^{O_d(1)}$ representative points
contained in a quadtree cell.
Since $T_A$ is balanced, it has height $O(\log|A|)=O(\log (1/\eps))$. Similarly, the height of $T_B$ is $O(\log (1/\eps))$. In other words, between any node in $A$ and any node in $B$ there exists a path in $T$ consisting of at most $O(\log(1/\eps))$ edges. 

We next prove the bound on the stretch between two points $a \in A$ and $b \in B$ that are $(\mu, \Delta)$-far. The proof follows the lines of \Cref{clm:nonsteiner-strip}.
Consider line $ab$ and let $a=c_1, c_2, \ldots c_k=b$ be the points on the path from $a$ to $b$ in $T$. For $1 \le i \le k$, let $c'_i$ be the intersection of $ab$ with a line\footnote{In higher dimensions, we consider the hyperplane orthogonal to $\theta$.} orthogonal to $\theta$ that passes through $c_i$; as the width of the strip is no more than $\e \Delta/\mu \le \e \norm{ab}$, we have $\norm{c'_i c_i} \le \e \norm{ab}$. By construction, every path in $T$ between a point in $A$ and a point in $B$ goes up the subtree $T_A$, passes through the root $a^*$ and goes down the tree $T_B$.  In other words, we have that for every  $i \in \{1, 2,\dots, k-1\}$ it holds that $\score_\theta(c_i) \le \score_\theta(c_{i+1})$ and $\score_\theta(c'_i) \le \score_\theta(c'_{i+1})$. In addition, $k = O(\log 1/\eps)$. We have $\norm{ab} = \sum_{1\le i\le k-1} \norm{c'_ic'_{i+1}}$.
Thus, the length of the path in $T$ is 
\begin{align*}
d_T(a,b) &= \sum_{1 \le i \le k-1}\norm{c_ic_{i+1}}\\
&\le \sum_{1 \le i \le k-1}(\norm{c'_ic'_{i+1}} + 2 \norm{c'_{i+1}c_{i+1}}) + \norm{c_1' c_1} &\text{by triangle inequality}\\
&\le \sum_{1 \le i \le k-1}(\norm{c'_ic'_{i+1}} + 2 \eps\norm{ab}) + \e \norm{ab}\\
&=\norm{ab}+O(\eps \log(1/\eps))\norm{ab}. 
\end{align*}

The above stretch argument guarantees that $T$ preserves path between any point in $X$ and any point in $Y$ up to a factor of  $1 + O(\eps \log(1/\eps))$. 
Applying the same degree reduction step for every strip in the strip partition and for every strip partition in the family $\zeta$, we obtain a tree cover with $O(1/\eps)$ trees and stretch $O(1+\eps\log(1/\eps))$. To complete the argument, we need to scale the parameters. 
Let $\EMPH{$\eps'$} \coloneqq O(\eps\log(1/\eps))$, so that the tree cover has stretch $1+\eps'$. The number of trees expressed in terms of $\eps'$ is $O(\frac{\log^d(1/\eps')}{\eps'})$, for $\eps < 1/16$. 

\paragraph{Degree \bm{$O_d(1)$}.} 
Let $A'  = \langle a^*, a_1, a_2, \ldots \rangle$ be the set of points in $A$, sorted in decreasing order with respect to $\score_\theta$. 
Recall that the direction of the major strip is $\theta$ and the direction of the minor strip is $\theta^\bot$. We build a binary tree $T_A$ rooted at $a_1$ as follows. 
Let the interval corresponding to $a_1$ be $[0, \eps \frac{\Delta}{2\mu})$. Recall that the width of the strip is $\eps \frac{\Delta}{2\mu}$. Let $\mathcal{I}$ be the set of active intervals, consisting of two elements: $[0, \eps \frac{\Delta}{4\mu})$, corresponding to the future left child of $a_1$ (if any) and $[\eps \frac{\Delta}{4\mu}, \eps \frac{\Delta}{2\mu})$ corresponding to the future right child of $a_1$ (if any). 
The elements of $\mathcal{I}$ form a partition of $[0, \eps \frac{\Delta}{2\mu})$ at all times.
Scan the points $a_2, a_3, \ldots$ in order and perform the following. 
Let $a_i$ be the currently scanned point and let $d_i$ be its distance from the left border of the strip. 
Go over all the intervals in $\mathcal{I}$ and see which one contains $d_i$. 
(Such an interval exists because $\mathcal{I}$ forms a partition of $[0, \eps \frac{\Delta}{2\mu})$.) 
Let $[l_i,r_i)$ be such an interval. 
Add $a_i$ at the corresponding place in the tree.  
Let $m_i \coloneqq (l_i + r_i)/2$. 
Create two new intervals: $[l_i, m_i)$ corresponding to the left child of $d_i$ and $[m_i, r_i)$, corresponding to the right child of $d_i$. 
Note that after this, $\mathcal{I}$ still forms a partition of $[0, \eps \frac{\Delta}{2\mu})$. This concludes the description of $T_A$. The tree $T_B$ is constructed analogously. 
Finally, the tree $T$ is obtained by attaching the roots of $T_A$ and $T_B$ as the left and right child of $a^*$. 

We next analyze the stretch. Consider two points $a \in A$ and $b \in B$ that are $(\mu, \Delta)$-far. Let $c_1 = a^*, c_2=a_1, c_3, \ldots, c_p$ be the path from $a^*$ (which is the root of $T$) to $a$ in $T$ and let $d_1 = a^*, d_2 = b_1, d_3, \ldots, d_q = b$ be the path from $a^*$ to $b$ in $T$. For two points $x$ and $y$, let $x = x_\theta \cdot \theta + x_\bot \cdot \theta^\bot$ and similarly $y = y_\theta\cdot \theta + y_\bot \cdot \theta^\bot$. Let $\norm{xy}_\theta = \abs{x_\theta - y_\theta}$ and $\norm{xy}_\bot = \abs{x_\bot - y_\bot}$. Using this notation, we observe that $\norm{ab}_\theta = \sum_{i=1}^{p-1} \norm{c_i c_{i+1}}_\theta + \sum_{i=1}^{q-1}\norm{d_id_{i+1}}_\theta$. The second observation is that $\sum_{i=1}^{p-1} \norm{c_i c_{i+1}}_\bot = O(\eps\Delta/\mu)$. This is because $\norm{c_i c_{i+1}}_\bot$ form a geometrically decreasing sequence. Similarly, $\sum_{i=1}^{q-1}\norm{d_id_{i+1}}_\bot = O(\eps\Delta/\mu)$. Using these two observations, we can upper bound the distance between $a$ and $b$ in $T$ as follows.
\begin{align*}
d_T(a,b) &= \sum_{i=1}^{p-1} \norm{c_i c_{i+1}} + \sum_{i=1}^{q-1}\norm{d_id_{i+1}}\\
&\le \sum_{i=1}^{p-1} \norm{c_i c_{i+1}}_\theta + \sum_{i=1}^{p-1} \norm{c_i c_{i+1}}_\bot + \sum_{i=1}^{q-1}\norm{d_id_{i+1}}_\theta +  \sum_{i=1}^{q-1}\norm{d_id_{i+1}}_\bot\\
&\le \norm{ab}_\theta + \sum_{i=1}^{p-1} \norm{c_i c_{i+1}}_\bot + \sum_{i=1}^{q-1}\norm{d_id_{i+1}}_\bot\\
&\le \norm{ab} + O(\eps\Delta/\mu) \\
&\le (1+\eps)\norm{ab}
\end{align*}

The argument for higher dimensions carries over almost exactly. The intervals used in the argument become $\mathbb{R}^{d-1}$-hypercubes. Consider a tree node $a$ and an interval $I_a\subset \mathbb{R}^{d-1}$ corresponding to it. We partition the interval $I_a$ into $2^{d-1}$ subintervals of twice the smaller side length. Those subintervals correspond to the children of $a$. To argue the stretch, we split the distance between points $a$ and $b$ in $T$ into two components: one along the vector $\theta$ and the remaining orthogonal part that lies in $\mathbb{R}^{d-1}$. The component along $\theta$ is at most $\norm{ab}$ and the component in $\mathbb{R}^d$ is at most $O(\eps)\norm{ab}$, due to the geometrically decreasing interval sizes.

Finally, we bound the diameter of each of the trees. Using analysis similar to the one used for the stretch, we conclude that the diameter of a tree corresponding to a single strip is at most $\norm{ab}(1+O(\eps)) \le 2\norm{ab}$. The trees of different major strips in a major strip partition are connected via a binary tree. As in \Cref{subsec:non-Steiner-d3}, the height of the binary tree is at most $\log\frac{4\mu d}{\eps}$. Hence, the overall degree is $2^{d-1} + 2$. The diameter of the tree is at most $2(2\Delta + \Delta \log\frac{4\mu d}{\eps}) \le 2\Delta \log \frac{16\mu d}{\eps}$.

\paragraph{Constant degree.} We next explain a tweak which leads to degree 5.
Instead of constructing a $2^{d-1}$-ary tree for each strip we can work with a binary tree. Tree $T_A$ is built as follows. Let $[0, \eps \frac{\Delta}{\mu d})^{d-1}$ be the interval corresponding to $a_1$. 
We assign \emph{level} to each node in the tree, ranging from $1$ to $d-1$. The level of $a_1$ is 1. The future children of $a_1$ are at level $2$. In general, the children of a node at level $i < d-1$ are at level $i+1$ and the children of a node at level $d-1$ are at level $1$. 
The set of active intervals $\mathcal{I}$ consists of $[0, \eps \frac{\Delta}{2\mu d})\times [0, \eps \frac{\Delta}{\mu d})^{d-2}$, corresponding to the left child of $a_1$ and $[\eps \frac{\Delta}{2\mu d}, \eps \frac{\Delta}{\mu d})\times [0, \eps \frac{\Delta}{\mu d})^{d-2}$. Once again, we maintain the property that $\mathcal{I}$ is a partition of $[0, \eps \frac{\Delta}{\mu d})^{d-1}$. Scan the points $a_2, a_3,...$ in that order and le
t $a_i$ be the currently scanned point and $d_i$ the $(d-1)$-dimensional vector of distances from each of the sides of the strip. 
Find the interval $I = [l_1, r_1) \times [l_2, r_2) \times \dots \times [l_{d-1}, r_{d-1})$ in $\mathcal{I}$ where $d_i$ belongs to and place $a_i$ at the corresponding place in the tree. 
Let $j \in \{1,2,\ldots, d \}$ be the level of $a_i$. 
Let $m_j \coloneqq (l_j + r_j) / 2$. Split $I$ into $I_l = [l_1, r_1) \times [l_2, r_2) \times \dots  \times [l_j, m_j) \times\dots \times [l_{d-1}, r_{d-1})$ corresponding to the left child of $a_i$ and $I_r = [l_1, r_1) \times [l_2, r_2) \times \dots  \times [m_j, r_j) \times\dots \times [l_{d-1}, r_{d-1})$ corresponding to the right child of $a_i$. Replace $I$ with $I_l$ and $I_r$ in $\mathcal{I}$. This concludes the description of the binary tree.

The stretch argument remains almost the same, except that $\sum_{i=1}^{p-1} \norm{c_i c_{i+1}}_\bot = O(d\eps\Delta/\mu)$, which is $d$ times larger than before. 
The reason is that every $d$ hops down the tree, we incur an additive stretch of $O(\eps\Delta/\mu)$ after which the additive stretch reduces by a factor of two. Using the same argument as before, we conclude that $d_T(a,b) \le (1+O(\eps d)) \norm{ab}$. By scaling the stretch, we get that the number of trees increases by a factor of $d^{d-1}$.

\newcommand{\Apices}{\ensuremath{\mathrm{Apices}}}
\newcommand{\lca}{\ensuremath{\mathrm{lca}}}
\newcommand{\cT}{\ensuremath{\mathcal{T}}}

\section{Application to Routing} \label{sec:route}
In this section, we show an application of our tree cover to compact routing scheme. First, we give some background on the problem. 
A \EMPH{compact routing scheme} is a distributed algorithm for sending messages or packets of information between points in the network. 
Specifically, a packet has an origin and it is required to arrive at a destination. 
Every node in the network contains a \EMPH{routing table}, which stores local routing-related information, and a unique \EMPH{label}, sometimes also called \EMPH{address}. 
In the beginning, the network is preprocessed and every node is assigned a routing table and a label.
Given a destination node $v$, routing algorithm is initiated at source $u$ and is given the label of $v$. Based on the local routing table of $u$ and the label of $v$, it has to decide on the next node $w$ to which the packet should be transmitted. 
More formally, the algorithm outputs the \EMPH{port number} leading to its neighbor $w$. 
Each packet has a message \EMPH{header} attached to it, which contains the label of the destination node $v$, but may also contain other helpful information. 
Upon receiving the packet the algorithm at node $w$ has at its disposal the local routing table of $w$ and the information stored in the header. 
This process continues until the packet arrives at its destination, which is node $v$. 
The \EMPH{stretch} of the routing scheme is the ratio between the distance packet traveled in the network and the distance in the original metric space.

We consider routing in metric spaces, where each among $n$ points in the metric corresponds to a network node. In the preprocessing stage, we choose a set of links that induces an \emph{overlay network} over which the routing must be performed. The goal is to have an overlay network of small size, whilst also optimizing  the tradeoff between the maximum storage per node (that is, the size of routing tables, labels, and headers) and the stretch.
In addition, one may try to further optimize the time it takes for every node to determine (or output) the next port number along the path, henceforth \emph{decision time}, and other quality measures, such as the maximum degree in the overlay network.

There are two different models, based on the way labels are chosen: \emph{labeled}, where the designer is allowed to choose (typically $\polylog(n)$) labels, and \emph{name-independent}, where an adversary chooses labels. Similarly, depending on who is choosing the port numbers, there is a \emph{designer-port} model, where the designer can choose the port number, and the \emph{fixed-port} model, where the port numbers are chosen by an adversary.   Our routing scheme works in the labeled, fixed-port model. 
For an additional background on compact routing schemes, we refer the reader to~\cite{Peleg00, TZ01, FG01, AGGM06}.

In this section, we prove Theorem~\ref{thm:routing}.

\subsection{Routing in trees} We first explain the interval routing scheme due to \cite{SK85}. Let $T$ be a given routed tree. We first preprocess the tree by performing a DFS on it and marking for every node $u$ the timestamp at which it got visited, $l_u$. For every node $u$, let $h_u$ be the maximum $l_w$ among the children $w$ of $u$. The label of node $u$ consists of $l_u$ and requires $\ceil{\log{n}}$ bits of storage. The routing table of node $u$ consists of the port number leading to its parent in $T$ (unless $u$ is a root), and for each child $w_i$ of $u$, the port number leading to $W_i$ together with $\langle l_{w_i}$, $h_{w_i}\rangle$. This requires $\deg_T(V) \cdot O(\log{n})$ bits. Specifically, it requires $O(\log{n})$ bits for trees of constant degree, which is the case for our construction. 
To route from some node $u$ to a destination $v$, the routing algorithm has routing table of $u$ and the label of $v$ at its disposal. For every child $w_i$ of $u$, if $l_v$ falls in the interval $\langle l_{w_i}$, $h_{w_i}\rangle$. If such a $w_i$ exists, the algorithm outputs the corresponding port to $w_i$ and otherwise the algorithm outputs port to the parent of $u$.
Note that in bounded degree trees the aforementioned routing algorithm needs to inspect only a constant number of entries in order to decide on the next port.

\subsection{Routing in Euclidean spaces}

To route in a Euclidean space, first construct a non-Steiner tree cover $\mathcal{T}$ with bounded degree, using Theorem~\ref{thm:boundedDegree}. 
The routing table of each point consists of its routing table in each of the trees in the cover, which takes $O_d(\eps^{-(d-1)} \log^2 1/\e \cdot \log{n})$ bits, since each tree is of a constant degree.
The label for each point consists of its label in each of the trees in $\mathcal{T}$, which overall takes $O_d(\eps^{-(d-1)} \log 1/\e \cdot \log{n})$ bits, together with an additional label of $O_d(\e^{-(d-1)} \log^2 1/\e \log n)$ bits described in the next section (``identifying a distance-preserving tree''). Overall this label takes $O_d(\eps^{-(d-1)} \log^2 1/\e \cdot \log{n})$ bits.
To route from a point $x$ to some other point $y$, the algorithm first identifies a tree in $\mathcal{T}$ that preserves the $xy$ distance up to a $(1+\e)$ factor: this step is described in the next section. After that, the routing algorithm proceeds on the single tree as described before.

\subsection{Identifying a distance-preserving tree}
Given two points $x$ and $y$ in $\R^d$, we now describe how to identify a tree in $\mathcal{T}$ that preserves the distance between $x$ and $y$ up to $1+\e$ stretch.
The total size of this label will be $O_d(\e^{-1} \log^2 1/\e \cdot \log n)$.

\paragraph{Review of tree cover construction.} We first recall the construction of Theorem~\ref{thm:NonSteiner}. We have a collection of compressed quadtrees $Q_i$ (for every $i \in [O_d(1)]$ and congruence classes $j \in [\ell]$ (where $\ell = O_d(\log 1/\e)$). For ease of notation, let \EMPH{$Q_{i,j}$} denote the tree obtained by starting with $Q_i$ and then contracting away all nodes except those at level $w$ for $w \equiv j \pmod \ell$. We refer to $Q_{i,j}$ as a \EMPH{contracted quadtree}. Notice that if $C$ is a cell in the contracted quadtree $Q_{i,j}$ with diameter $\Delta$, then the children of $C$ in $Q_{i,j}$ have diameter $O_d(\e\Delta)$. For every shift $i$ and congruence class $j$, we construct a set of trees as follows: for every cell $C$ in $Q_{i,j}$, we arbitrarily choose a set of $1/\e^{O_d(1)}$ representative points, one from each \emph{child cell} of $C$ in $Q_{i,j}$;  we construct a partial tree cover on the representative points; and we merge these partial tree covers together into a final set of trees.
Our proof of correctness guarantees that, for any pair of points $x$ and $y$, there is some contracted quadtree $Q_{i,j}$ and some cell $C$ in $Q_{i,j}$ with diameter $\Delta$, such that the two representative points $x'$ and $y'$  are $(\mu, \Delta)$-far. There is some tree in the partial tree cover of $C$ that preserves $||x'y'||$ up to a factor $1+\e$, and this tree corresponds to the tree in the final tree cover of Theorem~\ref{thm:NonSteiner} that preserves $||xy||$ up to a factor $1+O(\e)$.

In Section~\ref{sec:degreeLocal}, we constructed a tree cover in which each partial tree cover had bounded-degree. The construction is identical to that of Theorem~\ref{thm:NonSteiner}, except that we use a slightly modified construction for the partial tree cover on the representative points (modified from Section~\ref{sec:NonSteiner}).

In Section~\ref{sec:degreeGlobal}, we used the result of Section~\ref{sec:degreeLocal} to get a bounded-degree tree cover (proving Theorem~\ref{thm:boundedDegree}). The trees constructed in this section are in one-to-one correspondence with the trees constructed in Section~\ref{sec:degreeLocal}: if a tree $T$ in the cover of Section~\ref{sec:degreeLocal} preserves the distance between two points $x$ and $y$ up to a factor $1+\e$, the corresponding transformed tree $T'$ from Section~\ref{sec:degreeLocal} will preserve the distance up to a factor $1 + O(\e)$.

\bigskip
For simplicity, we describe how to identify a distance-preserving tree in the tree cover of Section~\ref{sec:degreeLocal}: for any $x$ and $y$, we will find a tree $T$ such that $\dist_T(x, y) \le (1+\e) \norm{xy}$. As described above, these trees are in one-to-one correspondence with the bounded-degree trees of Theorem~\ref{thm:boundedDegree} (which is the tree cover we actually use for routing).

Our labeling scheme will consist of a short label for each tree $Q_{i,j}$. For each $Q_{i,j}$, this label will let us identify a cell whose partial tree cover preserves the $||xy||$ distance (if such a cell exists), as well as the index of the corresponding distance-preserving tree. To construct this label for $Q_{i,j}$, we will need the following simple observation:

\begin{observation}
\label{obs:lca-preserving}
Let $x$ and $y$ be points in $X$, and let $Q_{i,j}$ be a contracted quadtree. Suppose there is a cell $\hat{C}$ of $Q_{i,j}$ such that $\hat{C}$ contains both $x$ and $y$, and the representatives $\hat{x}$ and $\hat{y}$ are $(\mu, \diam(\hat{C}))$-far. Then, in the \underline{smallest-diameter} cell $C$ that contains both $x$ and $y$, the representatives $x'$ and $y'$ are $(\mu, \diam(C))$-far. In other words, if we view $x$ and $y$ as leaves of $Q_{i,j}$, the lowest common ancestor of $x$ and $y$ guarantees that the representatives are  $(\mu, \diam(C))$-far.
\end{observation}

\subsubsection{Identifying a valid partial tree cover}
In this subsection, we describe a labeling scheme that lets us identify a distance-preserving tree in a \emph{partial} tree cover. 

\begin{lemma}
    \label{lem:far-label}
    Let $X \subset \R^d$ be a point set with diameter $\Delta$. For any constant $\mu = O_d(1)$, there is a labeling scheme with $O_d(1)$-bit labels, such that given the labels of any two points $x,y \in X$, we can either certify that $x$ and $y$ are $(\mu/4, \Delta)$-far or that they are not $(\mu, \Delta)$-far.
\end{lemma}
\begin{proof}
    Let $x \in R^d$ be a point with coordinates $x[1], \ldots, x[d]$.
    The label of $x$ consists of $d$ parts: for each coordinate $i \in \set{1, \ldots, d}$, the label stores difference between $x[i]$ and $\min_{x' \in X} x'[i]$ rounded to a multiple of $\Delta/(4 \mu d)$; that is, we store $\lfloor \frac{x[i]}{\Delta/(4 \mu d)} \rfloor$. Because the maximum such difference is $\Delta$ (because the diameter of $X$ is $\Delta$), the label takes $O_d(1)$ bits in total.

    Given the labels of any two points $x$ and $y$, we can compute, for each coordinate, an estimate of their difference within accuracy $\pm \Delta/(4\mu d)$. Thus, we can estimate $\ell_2$ distance between $x$ and $y$ within an accuracy of $\pm \Delta/(4\mu)$. If this estimated distance is at least $\Delta/2\mu$, the $\norm{xy} \ge \Delta/(4\mu)$, and so $x$ and $y$ are $(\mu/4, \Delta)$-far. Otherwise, if the estimated distance is smaller than $\Delta/(2\mu)$, we have $\norm{xy} < \Delta/\mu$, and so $x$ and $y$ are not $(\mu, \Delta)$-far.
\end{proof}

Notice that if a partial tree cover consists of $O(1/\e^{d-1})$ trees, one tree in the cover can be identified with $O_d(\log 1/\e)$ bits. We will allow these ``IDs'' of the trees to be fixed in advance.

\begin{lemma}
\label{lem:partial-label}
Let $\cT'$ be a $(\mu, \Delta)$-partial tree cover for a point set $X \subset \R^d$, constructed as in Section~\ref{sec:degreeLocal}, with $\mu = O_d(1)$.
Let $\operatorname{ID}:\cT' \to \set{0,1}^{k}$ be a function that maps trees to unique identifiers.
There is a labeling scheme for $X$ with $O_d(1/\e^{d-1} (\log 1/\e + k))$-bit labels, such that given the labels of any two points $x,y \in X$, we can either return $\operatorname{ID}(T)$ for some tree $T \in \cT'$ that preserves the distance $\norm{xy}$ up to a $1+O(\e)$ factor, or we can certify that $x$ and $y$ are not $(\mu, \Delta)$-far.
\end{lemma}
\begin{proof}
Recall from Sections~\ref{sec:NonSteiner} and \ref{sec:degreeLocal} that the construction of the partial tree cover $\cT'$ proceeds by constructing $O_d(1/\e^{d-1})$ \emph{major strip partitions} and $O_d(1)$ \emph{minor strip partitions}. 
The major strip partitions have width $\e \frac{\Delta}{2\mu} = \Theta_d(\e \Delta)$; thus, in each major strip partition, there are $1/\e^{O(d)}$ strips that contain points in $X$.
The minor strip partitions have width $\frac{\Delta}{2\mu} = O_d(\Delta)$; thus, in each major strip partition, there are $O_d(1)$ strips that contain points in $X$.
Every triple consisting of a major strip partition $P$ and two minor strips in the associated minor strip partition $P^\bot$ corresponds to some tree in the partial tree cover.

\medskip \noindent \textbf{Label.} For every point $x \in X$, the label consists of four parts:
\begin{itemize}
    \item For each of the $O_d(1/\e^{d-1})$ major strip partitions, store a $O_d(\log 1/\e)$-bit label identifying which strip in the major strip partition contains $x$.
    \item Similarly, for each of the $O_d(1)$ minor strip partitions, store a $O_d(1)$-bit label identifying which strip in the minor strip partition contains $x$.
    \item Store the $O_d(1)$-bit label of Lemma~\ref{lem:far-label}.
    \item For each of the $O_d(1/\e^{d-1})$ triples consisting of a major strip partition $P$ and two minor strips in the associated partition $P^\bot$, store the $k$-bit identifier (given by $\operatorname{ID}(\cdot)$) of the corresponding tree.
\end{itemize}

\medskip \noindent \textbf{Size.} The total size of all four parts is $O_d(1/\e^{d-1} (\log 1/\e + k))$.

\medskip \noindent \textbf{Label correctness.} Suppose we have the labels of two points $x, y \in X$.
First, we use the label of Lemma~\ref{lem:far-label} to determine either (1) $x$ and $y$ are $(\mu/4, \Delta)$-far, or (2) $x$ and $y$ are not $(\mu, \Delta)$-far. In the latter case, we are done; we have a certificate that $x$ and $y$ are not $(\mu, \Delta)$-far.

In the former case, we use the labels to check if there is some major strip partition $P$ such that the points $x$ and $y$ are in the same strip of $P$, and $x$ and $y$ are in different strips of the associated minor strip partition $P^{\bot}$.
If there is no such strip, then Claim~\ref{clm:good-strip-nonsteiner} implies that $x$ and $y$ are not $(\mu, \Delta)$-far, and we are done. Suppose there is such a strip. By the construction in Section~\ref{sec:degreeLocal}, this triple of major strip and minor strips corresponds to a tree $T$ in the partial tree cover $\cT'$. Claim~\ref{clm:nonsteiner-strip} implies that $\dist_T(a,b) \le \norm{ab} + \frac{\e\Delta}{\mu}$. As $a$ and $b$ are $(\mu/4, \Delta)$-far, we have $\dist_T(a,b) \le (1 + 4 \cdot \e)\norm{ab}$. Thus, we have identified a tree in $\mathcal{T}'$ that preserves the distance between $a$ and $b$ up to a $1 + O(\e)$ factor.
\end{proof}

\subsubsection{LCA Labeling Tools}
Before constructing our distance label, we need a preliminary result on LCA labeling. For any two vertices $x$ and $y$ in a tree $T$, let \EMPH{$\lca(x,y)$} denote the \EMPH{lowest common ancestor} of $x$ and $y$. For any vertex in the tree, we say its weight is the number of descendants. We say a vertex is \EMPH{heavy} if its weight is greater than half the weight of its parent, otherwise it is \EMPH{light}. For any vertex $x$,  let \EMPH{$\Apices[T, x]$}$ = \set{a_1, \ldots, a_{O(\log n)}}$ denote the parents of light ancestors of $x$. We remark on two important facts: (1) there are $O(\log n)$ vertices cells in $\Apices[T, x]$, and (2) the LCA of $x$ and $y$ is in $\Apices[T, x] \cup Apices[T, y]$. These facts are used in existing LCA labeling schemes. We will modify the labeling scheme of Alstrup, Halvorsen, and Larsen:

\begin{lemma}[Corollary 4.17 of \cite{AHL14}]
Let $T$ be a tree, and let $L : V(T) \to \set{0,1}^k$ be a function that indicates some predefined $k$-bit names for the vertices of $T$. There is a labeling scheme on the vertices of $T$ that uses $O(k \log n)$ bits, such that given labels of any two vertices $x$ and $y$, we can compute $L(\lca(x,y))$.
\end{lemma}
We will use a variant of their labeling scheme.

\begin{lemma}
\label{lem:lca}
Let $T$ be a tree. For every vertex $x$, let $L_x:V(T) \to \set{0,1}^k$ be a function that indicates some predefined $k$-bit names for the vertices of $T$. There is a labeling scheme on the vertices of $T$ that uses $O(k \log n)$ bits, such that given labels of two leaves $x$ and $y$, we can compute:
\begin{itemize}
    \item $L_x(\lca(x,y))$, if $\lca(x,y) \in \Apices[T, x]$
    \item $L_y(\lca(x,y))$, if $\lca(x,y) \in \Apices[T, y]$
\end{itemize}
If $\lca(x,y) \in \Apices[T,x] \cap \Apices[T,y]$, then we can compute both labels $(L_x(\lca(x,y)), L_y(\lca(x,y)))$.
\end{lemma}
\begin{proof}[Sketch]
We first review the labeling scheme of \cite{AHL14}. 
For every vertex $x$, the label of $x$ consists of two parts. 
The first part (cf.\ \cite[Corollary 4.17]{AHL14}) is just a lookup table: for every vertex $a \in \Apices[T,x]$, we record the $k$-bit name $L(a)$. The second part encodes information about the root-to-x path in the tree: in particular (cf. \cite[Lemma 4.13]{AHL14}), given labels for $x$ and $y$, we can use the second part of the label to detect whether $\lca(x,y)$ is in $\Apices[T,x]$ --- and to look up $L(\lca(x,y)$ in the lookup table, if $\lca(x,y)$ is in fact in $\Apices[T,x]$.

To obtain Lemma~\ref{lem:lca}, we simply change the lookup table in the label of $x$ to store $L_x(\cdot)$ instead of $L(\cdot)$. The proof of \cite{AHL14} guarantees that we can return $L_x(\lca(x,y))$ whenever $\lca(x,y) \in \Apices[T,x]$, and  symmetrically $L_y(\lca(x,y))$ whenever $\lca(x,y) \in \Apices[T,y]$.
\end{proof}

We are now ready to describe the label to identify a distance-preserving tree.

\subsubsection{Labeling scheme}

Let $\cT$ be the tree cover of Theorem~\ref{thm:boundedDegree}, of size $O_d(\e^{-(d-1)} \log (1/\e))$.

\paragraph{Label.} Let $Q_{i,j}$ be a contracted quadtree used in the construction of $\cT$. For each cell $C$ in $Q_{i,j}$, assign an arbitrary ordering to its $1/\e^{O_d(1)}$ children (so that we can specify a child of $C$ with $O_d(\log 1/\e)$ bits.) Let $x$ be a vertex in $X$, and treat $x$ as a leaf of $Q_{i,j}$. For every cell $C \in \Apices(Q_{i,j}, x)$, we define $L_x(C)$ to be a label consisting of three parts:
\begin{itemize}
    \item (L1) Store $O(\log 1/\e)$ bits to identify which child of $C$ is an ancestor of $x$.
    \item (L2) Store $O(\log 1/\e)$ bits to identify which child of $C$ is heavy (if there is a heavy child).
    \item (L3) Let $x'$ be the representative point for $x$. 
    Let $\cT'$ denote the partial tree cover at cell $C$, and for each tree $T' \in \cT$, define $\operatorname{ID}(T')$ to be the $O_d(\log 1/\e)$-bit identifier of the tree $T \in \cT$ of the final tree cover that contains $T'$.
    Store the label of $x'$ from Lemma~\ref{lem:partial-label} (using $\operatorname{ID}$): with this label, for any two points we can either find a tree in $\cT$ that preserves the distances of the representative points up to a factor $1 + O(\e)$, or we certify that the representative points are not $(\mu, \Delta)$-far.
\end{itemize}
Lemma~\ref{lem:lca} gives us an LCA label for $x$. For each contracted quadtree $Q_{i,j}$, store this label.

\paragraph{Size.} The label $L_x(C)$ consists of $O(\e^{-(d-1)} \log (1/\e))$ bits. Thus, Lemma~\ref{lem:lca} gives us labels of size $O(\e^{-(d-1)} \log (1/\e) \cdot \log n)$. There are $O_d(\log 1/\e)$ quadtrees $Q_{i,j}$, so the label size is $O_d(\e^{-1} \log^2 (1/\e) \cdot \log n)$ in total.

\paragraph{Decoding.} Suppose we have labels for $x$ and $y$. For each $Q_{i,j}$, we use Lemma~\ref{lem:lca} to find information about $C \coloneqq \lca(x,y)$. There are two cases:
\begin{itemize}
    \item \textbf{Case 1: \boldmath{$C$} is in \boldmath{$\Apices[Q_{i,j}, x] \cap \Apices[Q_{i,j}, y]$}.}  In this case, we have access to both $L_x(C)$ and $L_y(C)$. We use the (L1) parts of labels $L_x(C)$ and $L_y(C)$ to identify the two children $C_x$ and $C_y$ of $C$ that contain $x$ and $y$, respectively.
    \item \textbf{Case 2: (Without loss of generality) \boldmath{$C$} is only in \boldmath{$\Apices[Q_{i,j}, x]$}.} Let $C_y$ denote the child of $C$ that is an ancestor of $y$. Because $C$ is not the parent of a light ancestor of $y$, we know that the child $C_y$ is heavy. Use the (L2) part of $L_x(C)$ to identify the child $C_y$. As before, use the (L1) part of label $L_x(C)$ to identify the child $C_x$ that is an ancestor of $x$.
\end{itemize}
Having identified $C_x$ and $C_y$, we can now use the (L3) part of label $L_x(C)$ to determine whether
there is a tree that preserves the distance between the representatives of $C_x$ and $C_y$ up to a $1+O(\e)$ factor.
If there is such a tree, return it; otherwise, check the next $Q_{i, j}$.

By the proof of correctness of Theorem~\ref{thm:NonSteiner}, there is some contracted quadtree $Q_{i,j}$ with a cell in which the representatives of $x$ and $y$ are $(\mu, \Delta)$-far; further, our Observation~\ref{obs:lca-preserving} guarantees that it suffices to check only the LCA of $x$ and $y$ in each contracted quadtree. Thus, this process (iterating over all contracted quadtrees, and checking the LCA of each) will eventually find a quadtree cell in which the representatives are $(\mu, \Delta)$-far, and thus (by Lemma~\ref{lem:partial-label}), the (L3) part of the label will return a tree that preserves the distance of the representative points. By the proof of Claim~\ref{clm:stretch}, this tree preserves the distance between $x$ and $y$ up to a $1+O(\e)$ factor.

\paragraph{Acknowledgement.~}
Hung  Le  and  Cuong  Than  are  supported  by  the  NSF  CAREER  Award  No.  CCF-2237288 and an NSF Grant No. CCF-2121952.  Shay Solomon is funded by the European Union (ERC, DynOpt, 101043159).  Views and opinions expressed are however those of the author(s) only and do not necessarily reflect those of the European Union or the European Research Council.  Neither the European Union nor the granting authority can be held responsible for them.  Shay Solomon is also supported by the Israel Science Foundation(ISF) grant No.1991/1.  Shay Solomon and Lazar Milenkovi\'{c} are supported by a grant from the United States-Israel Binational Science Foundation (BSF), Jerusalem, Israel, and the United States National Science Foundation(NSF).

\small
\bibliographystyle{alphaurl}
\bibliography{refs}

\end{document}